\documentclass[11pt]{article}
\usepackage{ltexpprt}

\pdfoutput=1
\usepackage{latexsym}
\usepackage{amsmath}
\usepackage{amssymb}
\usepackage{enumerate}
\usepackage{thmtools,thm-restate}
\usepackage{graphicx}

\usepackage[vlined, algoruled, linesnumbered]{algorithm2e}

\newcommand{\proofend}{\hfill~$\Box$\\}
\newcommand{\ignore}[1]{}

\newcommand{\EdgeIsolated}[1]{$#1$-edge-cut}
\newcommand{\EdgeIsolatedIn}[1]{$#1$-edge-in}
\newcommand{\EdgeIsolatedOut}[1]{$#1$-edge-out}

\newcommand{\VertexIsolatedIn}[1]{$#1$-vertex-in}
\newcommand{\VertexIsolatedOut}[1]{$#1$-vertex-out}

\newcommand{\fEdgeIsoOut}[1]{{#1}EOut}

\newcommand{\fVertexIsoOut}[1]{{#1}VOut}

\begin{document}
\title{\bf Faster Algorithms for Computing Maximal $2$-Connected Subgraphs in Sparse Directed Graphs\footnote{A preliminary version of this paper appeared in~\cite{ChechikHILP17}}}

\author{
Shiri Chechik\footnote{Work partially done while visiting the Universit\`a di Roma ``Tor Vergata''. {This research was supported by the ISRAEL SCIENCE FOUNDATION
	(grant No. 1528/15).}}\\Tel Aviv University, Israel \\\texttt{schechik@cs.tau.ac.il}
\and
Thomas Dueholm Hansen\footnote{Work partially done while visiting the Universit\`a di Roma ``Tor Vergata''. Supported by the Carlsberg Foundation, grant no. CF14-0617.}\\Aarhus University, Denmark\\\texttt{tdh@cs.au.dk}
\and
Giuseppe F. Italiano\footnote{Partially supported by MIUR, the Italian Ministry of Education, University and Research, under Project AMANDA
(Algorithmics for MAssive and Networked DAta).}
\\Universit\`a di Roma ``Tor Vergata'', Italy\\\texttt{giuseppe.italiano@uniroma2.it}
\and
Veronika Loitzenbauer\footnote{Work partially done while visiting the Universit\`a di Roma ``Tor Vergata'' and the University of Michigan. 
The research leading to these results has received
funding from the Marshall Plan Foundation and the European Research Council under the
European Union’s Seventh Framework Programme (FP/2007-2013) / ERC Grant Agreement
no. 340506.}\\University of Vienna, Vienna, Austria\\\texttt{vl@cs.univie.ac.at}
\and
Nikos Parotsidis\\Universit\`a di Roma ``Tor Vergata'', Italy\\\texttt{nikos.parotsidis@uniroma2.it}}
\date{}

\maketitle

\begin{abstract}
Connectivity related concepts are of fundamental interest in graph
theory. The area has received extensive attention over four decades, but
many problems remain unsolved, especially for directed graphs.
A directed graph is $2$-edge-connected (resp., $2$-vertex-connected) if
the removal of any edge (resp., vertex) leaves the graph strongly
connected. In this paper we present improved algorithms for computing the maximal
$2$-edge- and $2$-vertex-connected subgraphs of a given directed graph. These 
problems were first studied more than 35 years ago, with $\widetilde{O}(mn)$ time
algorithms for graphs with $m$ edges and $n$ vertices being known since the late 
1980s. In contrast, the same problems for undirected graphs are known to be 
solvable in linear time.
Henzinger et al.\ [ICALP 2015] recently introduced $O(n^2)$ time algorithms for the 
directed case, thus improving the running times for dense graphs. Our new algorithms 
run in time $O(m^{3/2})$, which further improves the running times for sparse graphs.

The notion of $2$-connectivity naturally generalizes to $k$-connectivity for $k>2$.
For constant values of $k$, we extend one of our algorithms to compute 
the maximal $k$-edge-connected in time~$O(m^{3/2} \log{n})$,
improving again for sparse graphs the best known algorithm by
Henzinger et al.\ [ICALP 2015] 
that runs in $O(n^2 \log n)$ time. \end{abstract}

\makeatletter{}\section{Introduction}
Connectivity is one of the most well-studied notions in graph theory.
The literature covers many different aspects of connectivity related problems.
In this paper we study the problem of computing the maximal 
$k$-connected subgraphs of directed graphs.

\smallskip\noindent\textbf{Problem definition and related concepts.}
\emph{Strong connectivity.}
Let $G=(V,E)$ be a directed graph (digraph) with $m = \lvert E \rvert$ edges
and $n = \lvert V \rvert$ vertices.
The digraph $G$ is said to be \emph{strongly connected} if there is a directed path from
each vertex to every other vertex. 
The \emph{strongly connected components} (SCCs)
of~$G$ are its maximal strongly connected subgraphs.
Two vertices $u,v \in V$ are \emph{strongly connected} if they belong to the same
strongly connected component of~$G$.

\emph{2-edge connectivity.}
An edge of $G$ is a \emph{strong bridge} if its removal increases the number of
strongly connected components. Let $G$ be a strongly connected graph. We say that $G$ 
is $2$-edge-connected if it has no strong
bridges. Two vertices $v$ and $w$ are $2$-edge-connected if there are two
edge-disjoint paths from $v$ to $w$ and two edge-disjoint paths from $w$ to $v$.
A $2$-edge-connected component of $G$ is a maximal subset of vertices such that any
pair of distinct vertices is $2$-edge-connected. 
For a set of vertices~$C \subseteq V$ its induced subgraph $G[C]$ is  
a {\em maximal $2$-edge-connected subgraph} of $G$ if 
$G[C]$ is a $2$-edge-connected graph and no superset of $C$ has this property.
The $2$-edge-connected components of $G$ might be very different from the 
maximal $2$-edge-connected
subgraphs of $G$ because the two edge-disjoint paths between a pair
of vertices of a $2$-edge-connected component might use vertices that are
not in the $2$-edge-connected component. (See Figure~\ref{figure:notions} for an example.)

\begin{figure*}[t!]
	\begin{center}
		\includegraphics[trim={0cm 6.8cm 0cm 0cm}, clip=true, width=1.0\textwidth]{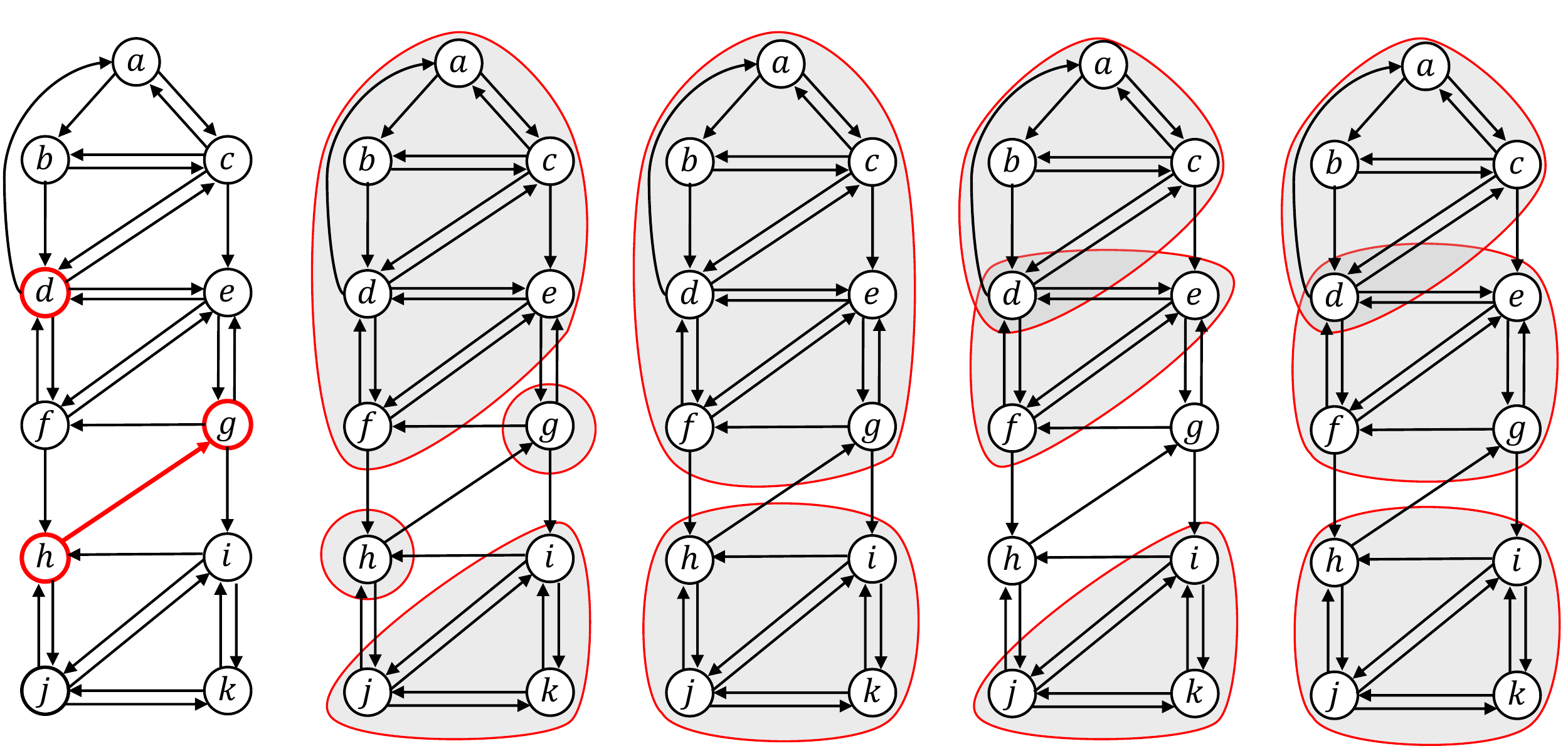}
	\end{center}	
		\hspace{1cm} (a) $G$ \hspace{1.7cm}  (b) $2\mathit{ECS}(G)$  \hspace{1.3cm}
		(c) $2\mathit{ECC}(G)$ \hspace{1.25cm}  (d) $2\mathit{VCS}(G)$ \hspace{1.35cm}  (e) $2\mathit{VCC}(G)$
		\caption{(a) A strongly connected digraph $G$; strong articulation points and strong bridges are shown in red. (b) The $2$-edge-connected subgraphs of $G$. (c) The $2$-edge-connected components of $G$. (d) The $2$-vertex-connected subgraphs of $G$.  (e) The $2$-vertex-connected components of $G$.
		}
		\label{figure:notions}
\end{figure*}

\emph{2-vertex connectivity.}
Analogous definitions can be given for $2$-vertex connectivity.
In particular, a vertex is a \emph{strong articulation point} if its removal
increases the number of strongly connected components of $G$.
Let $G$ be a strongly connected graph. The graph $G$ is $2$-vertex-connected if it has at least
three vertices and no strong articulation points.
Note that the condition on the minimum number of vertices disallows for degenerate
$2$-vertex-connected graphs consisting of two mutually adjacent vertices (i.e., two
vertices $v$ and $w$ and the two edges $(v,w)$ and $(w,v)$).
Two vertices $v$ and $w$ are $2$-vertex-connected if there are two internally
vertex-disjoint paths from $v$ to $w$ and two internally
vertex-disjoint paths from $w$ to $v$, i.e., the paths meet at $v$ and $w$ but not in-between (see also \cite{2VCB}).
A $2$-vertex-connected component of $G$ is a maximal subset of vertices such
that any distinct pair of vertices is $2$-vertex-connected.
For a set of vertices~$C \subseteq V$ its induced subgraph $G[C]$ is  
a {\em maximal $2$-vertex-connected subgraph} of $G$ if 
$G[C]$ is a $2$-vertex-connected graph and no superset of $C$ has this property. Note that the
$2$-vertex-connected components of $G$ might be very different from the maximal
$2$-vertex-connected subgraphs of $G$.

\emph{$k$-connectivity.}
The notions of $2$-edge and $2$-vertex connectivity extend naturally to $k$-edge and $k$-vertex connectivity.
Given a directed graph $G=(V,E)$, a set of edges $S$ is an \emph{edge cut} of size~$|S|$ if its removal increases the number of strongly connected components of~$G$.
A strongly connected graph
is $k$-edge-connected if it has no edge cut of size less than~$k$.
Two vertices $v$ and $w$ are $k$-edge-connected if there are $k$ edge-disjoint paths from $v$ to $w$ and $k$ edge-disjoint paths from $w$ to~$v$.
A $k$-edge-connected component of $G$ is a maximal subset of vertices such that any pair of distinct vertices is $k$-edge-connected.
For a set of vertices~$C \subseteq V$ its induced subgraph $G[C]$ is  
a {\em maximal $k$-edge-connected subgraph} of $G$ if 
$G[C]$ is a $k$-edge-connected graph and no superset of $C$ has this property.
A set of vertices $S$ is a \emph{vertex cut} of size $|S|$ if its removal increases the number of strongly connected components of $G$.
A strongly connected graph $G$ is $k$-vertex-connected if it has at least $k+1$ vertices and no vertex cut of size less than $k$.
Two vertices $v$ and $w$ are $k$-vertex-connected if there are $k$ internally vertex-disjoint paths from $v$ to $w$ and $k$ internally vertex-disjoint paths from $w$ to $v$.
A $k$-vertex-connected component of $G$ is a maximal subset of vertices such that any distinct pair of vertices is $k$-vertex-connected.
For a set of vertices~$C \subseteq V$ its induced subgraph $G[C]$ is  
a {\em maximal $k$-vertex-connected subgraph} of $G$ if 
$G[C]$ is a $k$-vertex-connected graph and no superset of $C$ has this property.

\emph{Undirected graphs.}
In undirected graphs a set of edges $S$ is an \emph{edge cut} of size $|S|$
if its removal increases the number of \emph{connected components} of the graph. 
An undirected connected graph is $k$-edge-connected if it has no edge cut of size 
less than $k$. The definitions of a vertex cut and of a $k$-vertex-connected graph
are analogous. The remaining definitions follow immediately from the definitions
for directed graphs.

Throughout the paper, we usually omit the word \emph{maximal} when referring to maximal 
$k$-edge- or $k$-vertex-connected subgraphs.

\smallskip\noindent\textbf{Our results.}
In this paper we present $O(m^{3/2})$ time algorithms
for computing the maximal $2$-edge-connected subgraphs and the maximal
$2$-vertex-connected subgraphs of a given directed graph with $m$~edges
and $n$~vertices. 
This is an improvement over the existing $O(n^2)$ time algorithms~\cite{2CC:HenzingerKL15}
whenever $m$ is $o(n^{4/3})$. 
The algorithm for $2$-edge-connected subgraphs is extended
to compute the maximal $k$-edge-connected subgraphs 
for any constant~$k \ge 2$ and 
runs in time~$O(m^{3/2} \log n)$, improving over the existing $O(n^2 \log n)$ time algorithm~\cite{2CC:HenzingerKL15}.
The maximal $k$-edge-connected (and $k$-vertex-connected)
subgraphs are defined for undirected graphs as they are for directed graphs.
We also show how to adjust
the algorithm
to compute the maximal $k$-edge-connected
subgraphs for undirected graphs in time $O((m+n\log n)\sqrt{n})$,
where $k$ is again viewed as a constant.
For the special case where $k=3$, the running time for computing 
the $3$-edge-connected 
subgraphs on undirected graphs is $O(m\sqrt{n})$.

\smallskip\noindent\textbf{Related work.}
In the literature the terms
``components'' and ``blocks'' have both been used to mean either $k$-connected components,
as defined above, or the maximal (induced) $k$-connected subgraphs; 
therefore we explicitly
use the term subgraphs for the latter in order to avoid further confusion.

\emph{Undirected graphs.}
It has been known for over 40 years how to compute the $2$-edge- and 
$2$-vertex-connected components of undirected graphs in linear time \cite{Tarjan72}.
While the $2$-edge-connected (resp., $2$-vertex-connected) components are equal
to the $2$-edge-connected (resp., $2$-vertex-connected) subgraphs in undirected graphs, 
this is no longer the case for $k>2$.
The first algorithm for computing the $3$-vertex-connected components in linear (in the number of edges)
time was by Hopcroft and Tarjan \cite{3-connectivity:ht}.
Later, Galil and Italiano~\cite{Galil:1991} reduced the computation of the 
$3$-edge-connected components to $3$-vertex-connected components, thus obtaining
a linear time algorithm for this case as well.
Kanevsky and Ramachandran~\cite{kanevsky1991improved} showed how to test whether a graph is $4$-vertex-connected in $O(n^2)$ time.
Over 20 years ago, Nagamochi and Watanabe~\cite{kECC:NW} presented an algorithm for computing the $k$-edge-connected components for $k>3$ in $O(m + k^2 n^2)$ time.
The best known algorithm for this problem runs in expected $\widetilde{O}(m+nk^3)$ time and was presented by Hariharan et al.~\cite{hariharan2007efficient}.
Their algorithm additionally computes a partial version of the Gomory-Hu tree~\cite{gomory1961multi}, that represents the edge-connectivity of the pairs whose edge-connectivity is less than $k$; the $k$-edge-connected components are contracted into singleton vertices in the tree.
Karger~\cite{Karger00} showed how to determine with high probability whether 
an undirected graph is $k$-edge-connected in $\widetilde{O}(m)$ time.
In a recent breakthrough, Kawarabayashi and Thorup~\cite{KawarabayashiT15} 
presented a deterministic algorithm with similar time bounds; 
Henzinger et al.~\cite{HenzingerRW17} improve the running time even beyond the randomized algorithm.
There is no study that explicitly considers the computation of the $k$-edge-connected
or the $k$-vertex-connected subgraphs of undirected graphs, however, the problem
can be reduced to the problem on directed graphs in a straightforward manner.
Furthermore, for undirected graphs the running time (which is implied 
by~\cite{gabow1995matroid}, see below) of the basic algorithm
for $k$-edge-connected subgraphs for constant~$k$ can be reduced to $O(n^2 \log n)$
by additionally maintaining a sparse certificate~\cite{EppsteinGIN97, 
NagamochiI92, Thorup07}.

\emph{$k$-connected components in digraphs.}
Very recently Georgiadis et al.~\cite{2ECB,2VCB} showed that the $2$-edge-connected
and the $2$-vertex-connected components of a directed graph can be computed in linear time.
Nagamochi and Watanabe~\cite{kECC:NW} gave an $O(k m n)$ time algorithm for computing
the $k$-edge-connected components in directed graphs.

\emph{$k$-edge-connected subgraphs in digraphs.}
A simple algorithm for computing the maximal $2$-edge-connected subgraphs is to
remove at least one strong bridge of a strongly connected component of the graph
and repeat on the resulting graph.
It is known since 1976 how to compute a strong bridge \cite{st:t} in $O(m+n\log n)$
time, and since 1985 in $O(m)$ time~\cite{GabowT85}, resulting in an $O(mn)$ time
algorithm for computing the $2$-edge-connected subgraphs of a directed graph.
Recently, Italiano et al.~\cite{Italiano2012} gave a linear time algorithm for
computing \emph{all} strong bridges of a directed graph in $O(m)$ time, of which
there can be $O(n)$ many.
A similar idea can be used to compute the $k$-edge-connected subgraphs.
In this case, in each iteration we remove the minimum edge cut of each strongly
connected component of the graph, if its size does not exceed $k-1$.
Since an edge cut of size $k$ can be computed in time $O(k m \log n)$
\cite{gabow1995matroid}, and in each iteration we disconnect at least one pair of
vertices, this algorithm runs in $O(kmn\log n)$ time.
Recently, Henzinger et al.~\cite{2CC:HenzingerKL15} presented an $O(n^2)$ time
algorithm for computing the $2$-edge-connected subgraphs of a directed graph
and an $O(n^2 \log n)$ time algorithm for the $k$-edge-connected subgraphs
for any constant~$k$. Their algorithm uses a sparsification technique introduced in
\cite{ChatterjeeH14, Henzinger1999} that can be used, under appropriate structural
properties, to replace a factor of $m$ in the running time of an algorithm by $n$.

\emph{$k$-vertex-connected subgraphs in digraphs.}
$2$-vertex-connected subgraphs were first studied in 1980 by Erusalimskii and Svetlov~\cite{erusalimskii1980bijoin}, but they did not analyze the running time of their algorithm.
Very recently, Jaberi \cite{2VCC:Jaberi2015} showed that their algorithm runs in $O(m^2 n)$ time and presented an $O(m n)$ time algorithm.
Prior to Jaberi, Makino \cite{Makino1988} gave an algorithm for computing the maximal
$k$-vertex-connected subgraphs of a directed graph in time $O(n \cdot \mathcal{S})$,
where $\mathcal{S}$ is the running time for computing a single vertex cut of size at
most  $k-1$.
Since one strong articulation point~\cite{Georgiadis10}, or even all the strong
articulation points~\cite{Italiano2012}, can be computed in linear time,
Makino's algorithm can be implemented so as to compute 
the $2$-vertex-connected subgraphs of a directed graph
in time $O(mn)$. Combined with Gabow's algorithm for identifying $k$-vertex
cuts~\cite{gabow2006using}, Makino's algorithm yields a running time of 
$O(m n \cdot (n + \min\{k^{5/2}, k n^{3/4}\}))$ for $k$-vertex-connected
subgraphs; an $O(k m n^2)$ time algorithm is already implied by combining it with~\cite{even1975algorithm}.
The recent algorithm of Henzinger et al.~\cite{2CC:HenzingerKL15} computes
the $2$-vertex-connected subgraphs in time~$O(n^2)$ and extends to the 
$k$-vertex-connected subgraphs for constant $k$ with a running time of $O(n^3)$.

\smallskip\noindent\textbf{Key Ideas.}
We next outline the main ideas behind our approach.
The basic algorithm for $2$-edge-connected subgraphs can be seen as maintaining
a partition of the vertices that is iteratively refined by identifying parts that
cannot be in the same $2$-edge-connected subgraph, which are then separated 
from each other
in the maintained partition. In the basic algorithm these parts are identified
by computing bridges and SCCs. The main technical
contribution of this work is a subroutine that can identify a ``small'' part
that can be separated from the rest of the graph by local depth-first searches
that, starting from one given vertex, explore only the edges in this small part
and a proportional number of edges outside of it.

For $2$-edge-connected subgraphs we call the subgraphs identified in this way
\emph{\EdgeIsolatedOut{1}} and \emph{\EdgeIsolatedIn{1}}
components\footnote{A similar notion called $2$-isolated set
was introduced in~\cite{2CC:HenzingerKL15}.}.
A \EdgeIsolatedOut{k} (resp., \EdgeIsolatedIn{k}) component of a vertex~$u$ is a subgraph (induced by some set of vertices) that contains~$u$ and has at most $k$ 
edges from (resp., to) the subgraph to (resp., from) the rest of the graph.
We start the searches for
these subgraphs from all vertices that have lost edges
since the last time bridges and SCCs were computed
and only recompute bridges and SCCs when no \EdgeIsolatedOut{1} or
\EdgeIsolatedIn{1} component with at most $\sqrt{m}$ edges exists\footnote{A similar
overall algorithmic structure was used in~\cite[Appendix~B]{2CC:HenzingerKL15} and,
for a different problem, e.g., in~\cite{ChatterjeeH14}.}.

The intuition for the local depth-first searches for edge connectivity
can be better understood in terms of maximum flow in uncapacitated graphs.
Assume there is a \EdgeIsolatedOut{1} component of a vertex~$u$. Since this
subgraph has at most one outgoing edge to the rest of the graph, the
vertex~$u$ can send at most one unit of flow to any vertex outside
of the subgraph. Thus if we find a path along which we can send one unit of flow
to some vertex outside of the subgraph and then look at the residual graph
given this flow, then there is no edge from the subgraph to the rest of the graph
in the residual graph. We find such a flow using depth-first search and then
use a second search to explore the subgraph that is still reachable from $u$
in the residual graph.

Finding $k-1$ paths to send flow out of a \EdgeIsolatedOut{(k-1)} component
is more difficult for $k > 2$. We show that one can exploit the properties of
depth-first search to find a set of $O(k)$ paths of which at least one of them
leaves the \EdgeIsolatedOut{(k-1)} component. As we have to do this for $k$ many
searches, each conducted in the residual graph after the previous search,
this yields an exponential dependence on $k$. For any constant $k > 2$
we compute the $k$-edge-connected subgraphs in time $O(m^{3/2} \log n)$ time,
where the additional factor of $\log n$ compared to $k=2$ is due to the 
increased cost of computing cuts with at most $k-1$ edges.

The notion of a \EdgeIsolatedOut{k} (resp., \EdgeIsolatedIn{k}) component of a vertex~$u$ is adjusted to vertex connectivity as follows.
A \VertexIsolatedOut{k} (resp., \VertexIsolatedIn{k}) component $S$ of a vertex~$u$ is a subgraph that contains~$u$ and at most $k$ vertices in the subgraph have 
edges from (resp., to) the subgraph to (resp., from) the rest of the graph.

For vertex connectivity some additional difficulties arise. First,
the $2$-vertex-connected subgraphs partition the edges rather than the vertices
(apart from degenerate cases),
i.e., when we find a strong articulation point and run our algorithm recursively on the subgraphs that it separates, the strong articulation point is included 
in each of these subgraphs.
Second, the intuition of flows and residual graphs cannot be applied directly;
instead, we let one depth-first search ``block'' specific vertices (those 
whose DFS subtree is adjacent to many edges) and let a second search ``unblock'' 
vertices such that it can explore the \VertexIsolatedOut{1} component but not 
the remaining graph.

It seems that the algorithm for computing $k$-edge-connected subgraphs can 
be extended to $k$-vertex-connected subgraphs by using the connection between 
flows and vertex connectivity shown in~\cite{EvenT75}. Additional details will 
appear in the full version of the paper.

\smallskip\noindent\textbf{Outline.}
After the preliminaries in Section~\ref{sec:preliminaries},
we present our algorithm for $2$-edge-connected
subgraphs in Section~\ref{sec:2edge}. 
We then describe the algorithm for $2$-vertex-connected subgraphs in Section~\ref{sec:2vertex}.
Finally, in Section~\ref{sec:kedge}, we state the results about $k$-edge-connectivity. 
 
\makeatletter{}
\section{Preliminaries}
\label{sec:preliminaries}
For a directed graph $G$ we denote by $V(G)$ its set of vertices and by $E(G)$
its set of edges. The reverse graph of a directed graph $G =(V,E)$, denoted by 
$G^R=(V,E^R)$, is the directed graph that results from $G$ after reversing the
direction of all edges. By $G \setminus S$ and $G \setminus Q$ we denote the 
graph~$G$ after the deletion of a set~$S$ of vertices and after the deletion of 
a set~$Q$ of edges, respectively.
We refer to the subgraph of~$G$ induced by the set of vertices~$S$ as $G[S]$.
Let $H$ be a strongly connected graph, or a strongly connected component of 
some larger graph.
We say that deleting a set of edges~$Q$ (resp., set of vertices~$S$) \emph{disconnects}
$H$, if $H \setminus Q$ (resp., $H\setminus S$) is not strongly connected.
Given a set of vertices~$C$, we say that a set of edges~$Q$ 
(resp., a set of vertices~$S$) disconnects $C$ from the rest of the 
graph if there is no pair of vertices $(x,y) \in C \times (V\setminus C)$
that are strongly connected in $G \setminus Q$ (resp., $G \setminus S$).
For the sake of simplicity, we write $S \subseteq G$, instead of $S \subseteq V(G)$,  to denote that a set of 
vertices~$S$ is a subset of the vertices of a graph~$G$.
We similarly write $Q \subseteq G$ instead of $Q\subseteq E(G)$, where $Q$ is a 
subset of the edges of the graph $G$.
Furthermore, we write $v\in G$ and $e\in G$ instead of $v \in V(G)$ and $e\in E(G)$, respectively.

We use the term tree to refer to a rooted tree with edges directed away from the root.
Given a tree~$T$, a vertex $u$ is an ancestor (resp., descendant)
of a vertex~$v$ if there is a directed path from $u$ to $v$ (resp., from $v$ to $u$)
in $T$. We denote by $T[u,v]$ the path from $u$ to $v$ in $T$.
We use $T(u)$ to denote the set of vertices that are descendants of $u$ in $T$.

There is a natural connection between edge cuts and maximum flow in unweighted graphs.
The maximum flow that can be sent from a source vertex $s$ to a target vertex $t$ in directed graphs with uncapacitated edges is equal to the number of edge-disjoint paths directed from $s$ to~$t$.
Therefore, the existence of a cut consisting of $k$ edges directed from a set of vertices $A$ to a set of vertices $B$ implies that the maximum flow that can be pushed from any vertex in $A$ to any vertex in $B$ is at most $k$.
Throughout the paper we implicitly use this connection between edge cuts and max flow.
We further assume that the reader is familiar with depth-first search (DFS), see, e.g., \cite{Tarjan72}.
 
\makeatletter{}\section{Maximal $2$-edge-connected subgraphs of a digraph}
\label{sec:2edge}

In this section we first show how to identify \EdgeIsolatedOut{1} components
that contain at most $\Delta$ edges in time proportional to $\Delta$.
Applied to the reverse graph, the same algorithm finds \EdgeIsolatedIn{1} components.
We then use this subroutine with $\Delta = \sqrt{m}$ to obtain an $O(m^{3/2})$
algorithm for computing the maximal $2$-edge-connected subgraphs of a given
directed graph.

\subsection{\EdgeIsolatedOut{1} and \EdgeIsolatedIn{1} components}
\label{sec:edgeIsolatedOne}

\begin{Definition}
Let $G = (V,E)$ be a digraph and $u \in V$ be a vertex. A \EdgeIsolatedOut{k} component of $u$ is a minimal subgraph $S$ of $G$ that contains $u$ and has at most $k$ outgoing edges to $G\setminus S$.
\end{Definition}

We similarly define a \EdgeIsolatedIn{k} component of $u$.

\begin{Definition}
	Let $G = (V,E)$ be a digraph and $u \in V$ be a vertex. A \EdgeIsolatedIn{k} component of $u$ is a minimal subgraph $S$ of $G$ that contains $u$ and has at most $k$ incoming edges from $G\setminus S$.
\end{Definition}

See Figure~\ref{figure:1-edge-out} for an example of a \EdgeIsolated{k} component and Figure~\ref{figure:1-edge-in} for an example of a \EdgeIsolatedIn{k} with $k=1$.
Note that $u$ may have more than one \EdgeIsolatedOut{k} (resp., \EdgeIsolatedIn{k}) component.
Also note that for $k' < k$, every \EdgeIsolatedOut{k'} component of $u$ is a \EdgeIsolatedOut{k} component of $u$ as well.
For the case when $k=1$, the outgoing (resp., incoming)
edge of a \EdgeIsolatedOut{1} (resp., \EdgeIsolatedIn{1}) component~$S$ is either
a strong bridge or an edge between strongly connected components of the graph.
Moreover, each $2$-edge-connected subgraph is either completely contained in~$S$
or in $G \setminus S$ (see also~\cite{2CC:HenzingerKL15}).

\begin{figure}[t!]
	\begin{center}
		\includegraphics[trim={0cm 6.8cm 20cm 0cm}, clip=true, width=.18\textwidth]{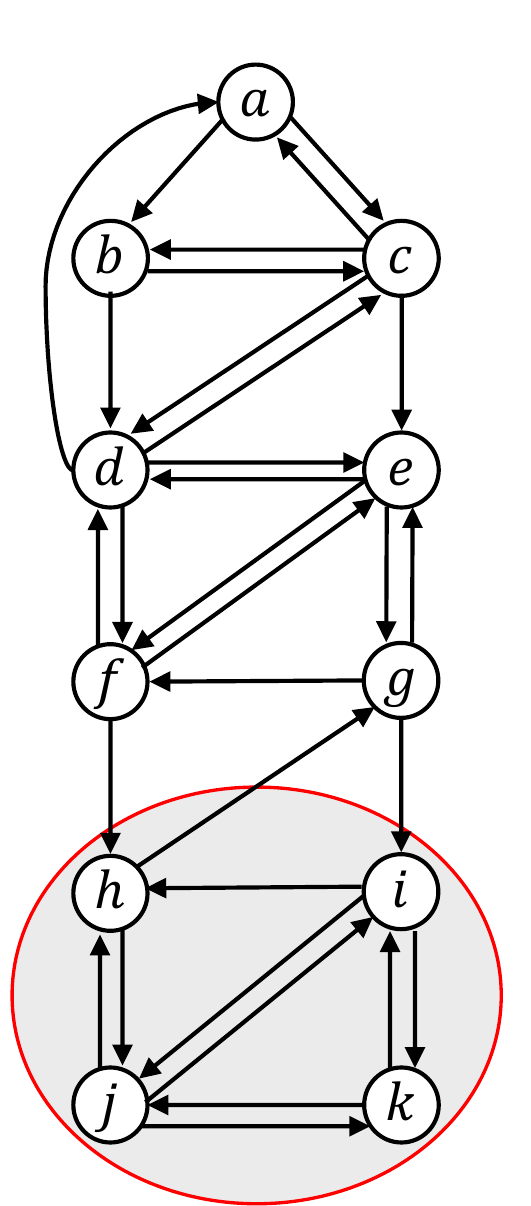}
	\end{center}
	\caption{An example of a \EdgeIsolatedOut{1} component of $j$.}
	\label{figure:1-edge-out}
\end{figure}

We next present an algorithm that takes as input a graph $G$, a vertex $u \in V(G)$, and a parameter $\Delta < m/2$, and that spends time at most $O(\Delta)$ to search for a \EdgeIsolatedOut{1} component of~$u$ in $G$. The algorithm may fail to find such a component, and we therefore prove the following guarantees about its outcome:
\begin{itemize}
\item
If $u$ has a \EdgeIsolatedOut{1} component with at most $\Delta$
edges, then the algorithm returns a \EdgeIsolatedOut{1} component for $u$ with at most
$2 \Delta$ edges.
\item
  If every \EdgeIsolatedOut{1} component of $u$ has more than $\Delta$ edges, then the algorithm may return a \EdgeIsolatedOut{1} component for $u$ with at most
$2 \Delta$ edges, but it may also return the empty set (i.e., fail to find a \EdgeIsolatedOut{1} component for $u$).
\end{itemize}
Note that by using exponential search in $\Delta$, the algorithm can find a \EdgeIsolatedOut{1} component for a given vertex $u$ in time that is linear in the number of edges of the smallest \EdgeIsolatedOut{1} component that contains $u$. For our purpose, however, it suffices to distinguish between small and large \EdgeIsolatedOut{1} components and only use one fixed choice of $\Delta$ (see Section~\ref{sec:2-edge-connected-strong-components}). We use the algorithm to quickly find a small \EdgeIsolatedOut{1} component~$S$, given a vertex~$u$ in~$S$.

\begin{figure}[t!]
	\begin{center}
		\includegraphics[trim={0cm 6.8cm 20cm 0cm}, clip=true, width=.18\textwidth]{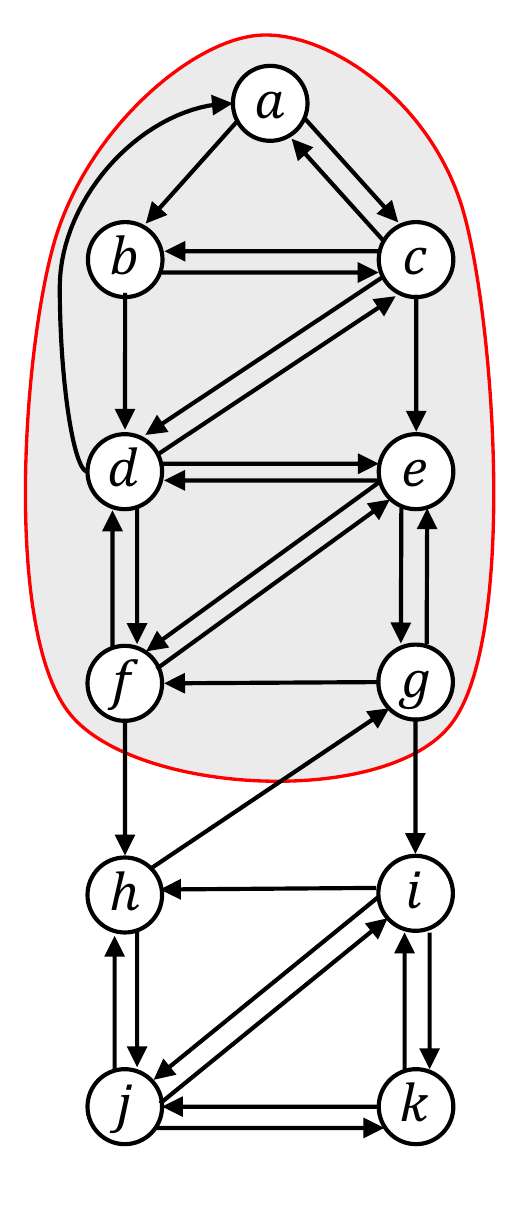}
	\end{center}	
	\caption{An example of a \EdgeIsolatedIn{1} component of $f$.}
	\label{figure:1-edge-in}
\end{figure}

For the rest of this section, we assume that the starting vertex $u$ can reach at
least $2\Delta+1$ edges.
Notice that if $u$ cannot reach $2\Delta+1$ edges, then the reachable subgraph from
$u$ defines a \EdgeIsolatedOut{0} component of $u$ containing at most $2\Delta$ edges.
In this case, the algorithm returns
this \EdgeIsolatedOut{0} component. We use exactly the same algorithm executed on the reverse graph for
\EdgeIsolatedIn{1} components and therefore only describe the algorithm
for \EdgeIsolatedOut{1} components.
First, we provide the following supporting lemmas. 

\begin{lemma}
	\label{lemma:2EISO-paths}
	Let $(x,y)$ be the outgoing edge of a \EdgeIsolatedOut{1} component $\fEdgeIsoOut{1}(u)$ of a vertex $u$.
	Then $u$ has a path to every vertex $v\in \fEdgeIsoOut{1}(u)$ that is contained entirely within the subgraph $\fEdgeIsoOut{1}(u)$.
	Moreover, $u$ has two edge-disjoint paths to $x$ within $\fEdgeIsoOut{1}(u)$.
\end{lemma}
 \begin{proof}
We begin by showing that $u$ has a path to every vertex $v\in \fEdgeIsoOut{1}(u)$ that is contained entirely within the subgraph $\fEdgeIsoOut{1}(u)$.
	Assume for the sake of contradiction that there is a set of vertices $C \subset 
	\fEdgeIsoOut{1}(u)$ such that the vertices of~$C$ are unreachable from $u$ in 
	$\fEdgeIsoOut{1}(u)$.
	Then there is no edge $(w,z)$ with $w\in \fEdgeIsoOut{1}(u) \setminus C$ and 
	$z \in C$ and thus the only possible outgoing edge from $\fEdgeIsoOut{1}(u) \setminus C$ is $(x,y)$.
	Thus, $\fEdgeIsoOut{1}(u) \setminus C$ is a \EdgeIsolatedOut{1} component of $u$, which contradicts the minimality of $\fEdgeIsoOut{1}(u)$.
	
	We now show that $u$ has two edge-disjoint paths to $x$ in $\fEdgeIsoOut{1}(u)$.
	First, we note that all simple paths from $u$ to $x$ contain only vertices in $\fEdgeIsoOut{1}(u)$ since there is no edge $(x',y') \not= (x,y)$ leaving $\fEdgeIsoOut{1}(u)$.
	Assume, for the sake of contradiction, that all paths from $u$ to $x$ in $\fEdgeIsoOut{1}(u)$ share a common edge $(w,z)$. 
	Then, $u$ does not have a path to $z$ in $\fEdgeIsoOut{1}(u) \setminus (w,z)$.
	Let $C \subset \fEdgeIsoOut{1}(u)$ be the set of vertices that become unreachable from $u$ in $\fEdgeIsoOut{1}(u) \setminus (w,z)$. (Notice that $|C| \geq 1$ since $z \in C$.) 
	Clearly, there is no edge $(w',z')$ such that $w' \in V(\fEdgeIsoOut{1}(u)) \setminus C$ and $z' \in C$.
	Hence, the only outgoing edge from $\fEdgeIsoOut{1}(u) \setminus C$ is $(w,z)$.
	Thus, $\fEdgeIsoOut{1}(u) \setminus C$ is a \EdgeIsolatedOut{1} component of $u$, which again contradicts the minimality of $\fEdgeIsoOut{1}(u)$.\proofend
 \end{proof}

Our algorithm starts a DFS traversal $F_1$ from $u$. 
We charge to a visited vertex its outgoing edges that were discovered by~$F_1$.
We stop $F_1$ when the number of traversed edges reaches $2\Delta+1$.
Let $T$ be the DFS tree constructed by the DFS traversal.
We define the weight of a vertex $v$, denoted by $w(v)$, to be the total number of
edges charged to the descendants of $v$ in $T$ (including $v$).
Assume $u$ is contained in a \EdgeIsolatedOut{1} component $C$ with at most 
$\Delta$ edges. Then the DFS has to leave $C$ via its only outgoing edge in order to 
reach more than $\Delta$ edges.
Note that for any vertex $v \ne u$ whose DFS subtree only explores edges inside $C$
we have $w(v) < \Delta$. The following two lemmata show that the vertices with 
$w(v) \ge \Delta$ form a path from $u$ to a vertex outside of $C$ that 
we can then use to block the outgoing edge such that the second traversal
explores exactly $C$.

\begin{lemma}
	Let $\fEdgeIsoOut{1}(u)$ be a \EdgeIsolatedOut{1} component of $u$ with 
	outgoing edge $(x,y)$ such that $|E(\fEdgeIsoOut{1}(u))| \leq \Delta$, 
	and let $T$ be a DFS tree of a DFS traversal from $u$ that visits
	$2\Delta+1$ edges. Then $w(v) \ge \Delta$ for each
        vertex $v$ on the path from $u$ to $y$ in $T$, i.e., $v \in T[u,y]$, and $w(v)<\Delta$ for each $v \in \fEdgeIsoOut{1}(u) \setminus T[u,x]$.
	\label{lemma:2EISO-reversed-edges}
\end{lemma}
\begin{proof}
Since $|E(\fEdgeIsoOut{1}(u))| \leq \Delta$, the DFS traversal has to 
visit vertices outside of $\fEdgeIsoOut{1}(u)$ to reach more than $\Delta$
edges. The DFS traversal can leave $\fEdgeIsoOut{1}(u)$ only by using 
the edge $(x,y)$ and it can do so only once. As the DFS traversal visits at least 
$2\Delta+1$ edges, it visits at least $\Delta$ edges in the 
subtraversal from $y$ (i.e., subsequent to exploring the DFS tree edge $(x,y)$).
Therefore, for each $v \in T[u,y]$ it holds that $w(v)\geq \Delta$.
Moreover, any subtraversal that does not visit vertices outside of
$\fEdgeIsoOut{1}(u)$ cannot include more than $\Delta$ edges.
None of the subtraversals from vertices $v \in \fEdgeIsoOut{1}(u) \setminus T[u,y]$
can visit vertices outside of $\fEdgeIsoOut{1}(u)$, since these vertices are 
either visited after the edge $(x,y)$, or their subtraversal can not visit $x$ 
(i.e., the vertex~$v$ cannot reach $x$ without using backedges w.r.t.\ the DFS tree).
Thus, for each vertex $v \in \fEdgeIsoOut{1}(u) \setminus T[u,x]$, it holds that 
$w(v)<\Delta$.\proofend
\end{proof}

\begin{lemma}
	\label{lemma:reversed-path}
	Let $F$ be a DFS traversal that visited $2\Delta+1$ edges and let $T$ be the
	DFS tree generated by~$F$.
	The edges $e = (x,y) \in T$ with $w(y) \geq \Delta$ form a path in $T$.
\end{lemma}
\begin{proof}
	Assume by contradiction that there are two distinct tree edges $e_1=(x_1,y_1)$ and 
	$e_2=(x_2,y_2)$ with $w(y_1) \geq \Delta$ and $w(y_2)\geq \Delta$ 
	that do not have an ancestor-descendant relation in $T$ (i.e., $y_1$ is not an 
	ancestor of $x_2$ and $y_2$ is not an ancestor of $x_1$).
	Since also the edges $(x_1,y_1)$ and $(x_2,y_2)$ are visited by $T$, this 
	contradicts the fact that the traversal visited $2\Delta+1$ edges.
	Therefore, all edges $e = (x,y) \in T$ with $w(y) \geq \Delta$ form a path in $T$.\proofend
\end{proof}

After the execution of the first DFS $F_1$, by Lemma~\ref{lemma:reversed-path},
there is a path $P$ of $T$ such that we have $w(y) \geq \Delta$ for every edge
$e=(x,y)$ of $P$. We call this path the \emph{heavy path} of $F_1$, and the edges
contained in the heavy path the \emph{heavy edges} of $F_1$.
Note that (1) the heavy path has to leave a \EdgeIsolatedOut{1} component of $u$
with at most $\Delta$ edges for the search to reach more than $\Delta$ edges
and (2) the heavy path cannot enter the component again after 
leaving it because the subtree of any incoming edge of the component cannot
contain $\Delta$ or more edges as the only outgoing edge of the component
was already used.
We construct the residual graph $G'$ formed from~$G$ by reversing the
direction of the heavy edges of~$F_1$.
The residual graph will be used as follows.
If there exists a \EdgeIsolatedOut{1} component $\fEdgeIsoOut{1}(u)$ of $u$ containing
at most $\Delta$ edges, then the heavy path~$P$ can be interpreted as sending
one unit of flow out of $\fEdgeIsoOut{1}(u)$ and in the residual graph with
respect to this flow no additional unit of flow can be sent out of
$\fEdgeIsoOut{1}(u)$. That means that no other search from~$u$ is able to have an
outgoing path from $\fEdgeIsoOut{1}(u)$.
Next, we execute a second traversal $F_2$ from $u$ (not necessarily a depth-first search) on $G'$.
We show that if there exists a \EdgeIsolatedOut{1} component $\fEdgeIsoOut{1}(u)$ of $u$ containing
at most $\Delta$ edges, this second traversal has two main properties: (\textit{i}) it never visits edges outside of $G'[V(\fEdgeIsoOut{1}(u))]$, and (\textit{ii}) it visits all the edges in $G'[V(\fEdgeIsoOut{1}(u))]$.
Whenever $F_2$ traverses more than $\Delta$ edges, we terminate the search and conclude that any \EdgeIsolatedOut{1} component of $u$ contains more than $\Delta$ edges.
	
\begin{lemma}
	Let $G'$ be the residual graph obtained from $G$ by reversing the direction of the
	heavy edges of $F_1$.
	The traversal $F_2$ reaches at most $\Delta$ edges in $G'$ if and only if there
	exists a \EdgeIsolatedOut{1} component $\fEdgeIsoOut{1}(u)$ of $u$ containing at
	most $\Delta$ edges.
	Moreover, if $F_2$ traverses at most $\Delta$ edges, then the subgraph in $G$
	induced by the vertices traversed by $F_2$ defines $\fEdgeIsoOut{1}(u)$.
	\label{lemma:determin-1-edge-out-component}
\end{lemma}
\begin{proof}
	Let us first assume that there exists a \EdgeIsolatedOut{1} component 
	$\fEdgeIsoOut{1}(u)$ of $u$ that contains at most $\Delta$ edges and has one 
	outgoing edge $(x,y)$.
	By Lemma \ref{lemma:2EISO-reversed-edges}, the edge $(x,y)$ is reversed in the 
	residual graph $G'$. Moreover, the lemma implies that 
	no incoming edge to $\fEdgeIsoOut{1}(u)$ is 
	reversed in $G'$ because each incoming edge $(v, z)$ either has $w(z) < \Delta$
	or $z \in T[u, x]$; in the latter case $(v, z)$ cannot be a DFS tree edge
	as $T[u, x]$ is contained in $\fEdgeIsoOut{1}(u)$, thus $v \not\in T[u, x]$, and 
	hence $(v, z)$ being a DFS tree edge would generate a cycle in the DFS tree.
	Thus, $G'[V(\fEdgeIsoOut{1}(u))]$ has no outgoing edges to $G'[V(G) \setminus V(\fEdgeIsoOut{1}(u))]$.
	Therefore, $F_2$ cannot visit more than $\Delta$ edges.
	We now show that $F_2$ visits all vertices in $G'[V(\fEdgeIsoOut{1}(u))]$ using only paths internal to $G'[V(\fEdgeIsoOut{1}(u))]$.
	Notice that this does not trivially follow from Lemma \ref{lemma:2EISO-paths} since we are operating on the residual graph $G'$, where the direction of some edges of $\fEdgeIsoOut{1}(u)$ is reversed.
	Assume by contradiction that $u$ cannot visit all vertices in $G'[V(\fEdgeIsoOut{1}(u))]$.
	Then, there is a set of vertices $C \subset V(\fEdgeIsoOut{1}(u))$ that has no 
	incoming edge from $V(\fEdgeIsoOut{1}(u)) \setminus C$ in the residual graph $G'$.
	By Lemma~\ref{lemma:reversed-path} the edges that are reversed in the residual 
	graph~$G'$ form a path~$P$ in the DFS tree of $F_1$.
	The path~$P$ contains an incoming edge to $C$ in $G$ since otherwise 
	$\fEdgeIsoOut{1}(u) \setminus C$ is a \EdgeIsolatedOut{1} component of $u$, 
	contradicting the minimality of $\fEdgeIsoOut{1}(u)$.
	Since $C$ has no incoming edges from $\fEdgeIsoOut{1}(u) \setminus C$ in $G'$, 
	we have that $P$ has no outgoing edges from $C$ to $\fEdgeIsoOut{1}(u) \setminus C$.
	Therefore $T[u, y] \in P$ implies $x \in C$.
	Since $P$ does not enter $\fEdgeIsoOut{1}(u)$ after leaving through $(x, y)$,
	only one edge incident to $C$ was reversed in $G'$. As there is
	no edge incident to $C$ in $G'$, this is a contradiction to 
	Lemma~\ref{lemma:2EISO-paths}, which says that $u$ has two edge-disjoint paths 
	to~$x$. Hence no such set~$C$ exists and $F_2$ traverses all vertices of 
	$\fEdgeIsoOut{1}(u)$.

	Now we show the opposite direction.
	Assume that $F_2$ visits at most $\Delta$ edges in the residual graphs. 
	We will show that there exists a \EdgeIsolatedOut{1} component $\fEdgeIsoOut{1}(u)$ 
	of $u$ that contains at most $\Delta$ edges and 
	that is given by the subgraph induced  
	by the vertices traversed by $F_2$.
	Let $C$ be the subgraph that $F_2$ traversed in the residual graph.
	Then $C$ has no outgoing edges in $G'$, since otherwise their neighbors would 
	 also be traversed by $F_2$.
	Since $F_1$ visited $2\Delta+1$ edges, there is at least one edge $e^*$ incoming to $C$ in $G'$ that was reversed.
	Note that there cannot exist more than one incoming edge to $C$  in $G'$ that was reversed after $F_1$, since that would imply the existence of an outgoing edge from $C$ since the set of reversed edges forms a path by Lemma~\ref{lemma:reversed-path}.
	Hence $u$ has no path to any of the vertices in $V \setminus C$ in the residual graph~$G'$, and has only one outgoing edge in the original graph~$G$.
	Therefore, after restoring the reversed edges, $C$ forms a \EdgeIsolatedOut{1}
	component of $u$ that contains at most $\Delta$ edges, with the only outgoing 
	edge being~$e^*$.
	Notice that the vertices of $C$ were all traversed by $F_2$.
	It remains to show that there is no \EdgeIsolatedOut{1} component $\fEdgeIsoOut{1}'(u)$ of $u$ with one outgoing edge $(x',y')$ and such that $\fEdgeIsoOut{1}'(u) \subset \fEdgeIsoOut{1}(u)$.
	Assume by contradiction that there exists such a component.
	By Lemma~\ref{lemma:2EISO-reversed-edges} the traversal $F_1$ reversed $(x',y')$, and there is no other outgoing edge from $\fEdgeIsoOut{1}'(u)$ in the residual graph. 
	Therefore, $F_2$ cannot visit vertices outside $\fEdgeIsoOut{1}'(u)$.
	A contradiction to the fact that $F_2$ visited all the edges and vertices in $\fEdgeIsoOut{1}(u)$.\proofend
	\end{proof}
	
	\begin{procedure}[t!]
		\DontPrintSemicolon
		\KwIn{Digraph $G=(V,E)$, a vertex $u$, and an integer $\Delta$}
		\KwOut{Either a \EdgeIsolatedOut{1} component of $u$ with at most $2\Delta$ edges
		or $\emptyset$; if $\emptyset$ is returned, then every \EdgeIsolatedOut{1} component that contains $u$ has more than $\Delta$ edges}
		\BlankLine
		
		Execute DFS $F_1$ from $u$ for up to $2\Delta+1$ edges\;
		Let $S_1$ be the vertices reached by $F_1$
		
		\If{ $F_1$ cannot reach $2\Delta+1$ edges}{\Return $G[S_1]$ as \EdgeIsolatedOut{1} component of $u$}
		\Else{

			Let $P$ be the heavy path of $F_1$\;
			Let $G'$ be $G$ after reversing the direction of the edges of~$P$
			
			Execute DFS $F_2$ from~$u$ on $G'$ for up to $\Delta+1$ edges\;
			Let $S_2$ be the vertices reached by $F_2$
			
			\If{$F_2$ cannot reach $\Delta+1$ edges} {
				\Return $G[S_2]$ as \EdgeIsolatedOut{1} component of $u$
			}
			\Else{\Return $\emptyset$}
		}
		\caption{1EdgeOut($G$, $u$, $\Delta$)}
		\label{proc:1EdgeOut}
	\end{procedure}
	
	Recall that we assumed in the beginning of this section that $u$ reaches at least $2\Delta+1$ edges. If this is not satisfied, we return the set of reachable
	vertices from $u$, which is a \EdgeIsolatedOut{0} component of $u$ with 
	at most $2\Delta$ edges. Otherwise the first DFS search $F_1$
	is able to visit $2\Delta+1$ edges. 
	After the execution of the second traversal $F_2$ on the residual graph $G'$, we can answer whether there exists a \EdgeIsolatedOut{1} component of $u$ with at most $\Delta$ edges, as shown in Lemma \ref{lemma:determin-1-edge-out-component}.
	The pseudocode of our algorithm is illustrated in Procedure~\ref{proc:1EdgeOut}.
	The following lemma summarizes the result of this section.
	
	\begin{lemma}
		Procedure~\ref{proc:1EdgeOut} computes a \EdgeIsolatedOut{1} \textup{(}resp.,
		\EdgeIsolatedIn{1}\textup{)} component of $u$ with at most $2\Delta$ edges or
		decides that there is no \EdgeIsolatedOut{1} \textup{(}resp.,
		\EdgeIsolatedIn{1}\textup{)} component of $u$ with at most $\Delta$ edges.
		Moreover, Procedure~\ref{proc:1EdgeOut} runs in $O(\Delta)$ time.
	\end{lemma}
	
	\subsection{Computing the $2$-edge-connected subgraphs.}
	\label{sec:2-edge-connected-strong-components}
	
	Let $G = (V,E)$ be a digraph.
	A straightforward algorithm for computing the $2$-edge-connected subgraphs is to recursively remove, from~$G$, one strong bridge of each strongly connected component of $G$ until no strong bridges can be found.
	In each recursive call at least one vertex becomes disconnected from the rest of the graph.
	Since computing the strongly connected components and one strong bridge (or all strong bridges) of a digraph can be done in linear time, this simple algorithm runs in $O(m n)$ time. 	
	\begin{algorithm}[t!]
		\DontPrintSemicolon
		\KwIn{		A strongly connected digraph $G=(V,E)$ and
		a list of vertices $L$ (initially $L = V$)}
		\KwOut{The $2$-edge-connected subgraphs of $G$}
		\BlankLine
		Let $m_0$ be number of edges of initial graph\;
		\If{$G$ has no strong bridge\label{alg-line:Bridges}}{
			\Return $\{G\}$ as $2$-edge-connected subgraph
		\label{alg-line:output-2ec}
		}		
		\While{$L \not = \emptyset$ and $G$ has more than $2 \sqrt{m_0}$ edges}
		{
			Extract a vertex $u$ from $L$
			
			$S \gets $\ref{proc:1EdgeOut}($G$, $u$, $\sqrt{m_0}$) \label{alg-line:search-components1}\;
			$S^R \gets $\ref{proc:1EdgeOut}($G^R$, $u$, $\sqrt{m_0}$)
			\label{alg-line:search-components2}\;
			
			If either $S$ or $S^R$ is not empty, remove from $G$ all edges incident
			to one non-empty set of $S$ and $S^R$ and add their endpoints to $L$
			\label{alg-line:found-component}
		}
		
		Compute SCCs $C_1, \ldots, C_c$ of $G$\label{alg-line:sccs1}\;
		$U \gets \emptyset$\;
		\ForEach{$C_i, 1 \leq i \leq c$}{
			Remove one strong bridge from $C_i$ (if one exists)\label{alg-line:strong-bridges-in-sccs}\;
			Recompute SCCs and delete the edges between them\label{alg-line:sccsdel}\;
			\ForEach{SCC $C'$}{
				Initialize $L'$ with the vertices of $C'$ that are endpoints of newly
				deleted edges\;
				$U \gets U \cup 2ECS(C',L')$\label{alg-line:recursion-big-sets}\;
			}
		}
		\Return $U$
		\caption{\textsf{$2ECS(G, L)$ }}
		\label{alg:2ECSC}
	\end{algorithm}

	In our algorithm we build on the simple algorithm described above.
	The high-level idea of our approach
	is to (a) find subgraphs with at most $\sqrt{m}$ edges that are not 2-edge-connected 
	to the rest of the graph in total time $O(m \sqrt{m})$
	and by this (b) limit the maximum recursion depth to $\sqrt{m}$
	by only making recursive calls
	when large subgraphs will be disconnected from each other or the remaining graph
	has at most $O(\sqrt{m})$ edges.
	This is done as follows.
	We use the terms small and large components to refer to subgraphs that
	contain at most and more than $\sqrt{m}$ edges, respectively.
	We first identify all the small components that can be disconnected from the
	rest of the graph by a single edge deletion.
	In each recursive call of the algorithm we maintain a list $L$ of vertices
	for which we want to identify small \EdgeIsolatedOut{1} and \EdgeIsolatedIn{1}
	components. Initially, we set the list $L$ to contain all vertices in order to
	find all small components that can be separated by at most one edge.
		We search for such small subgraphs using the algorithm from
	Section~\ref{sec:edgeIsolatedOne}.
	We compute \EdgeIsolatedIn{1} components by executing $1EdgeOut(G^R,u,\sqrt{m})$,
	where $G^R$ is the reverse graph of $G$.
	Whenever we find a small \EdgeIsolatedOut{1} or \EdgeIsolatedIn{1} component, we
	remove all its incident edges and search for more small \EdgeIsolatedOut{1} or
	\EdgeIsolatedIn{1} components in the remaining graph.
	We do that by inserting the endpoints of the deleted edges into the list $L$.
	If, on the other hand, we cannot find new small components, we conclude that
	either the remaining graph is 2-edge-connected or 
	there are at least two large sets of vertices that will get disconnected by
	either recomputing SCCs or by the removal of a strong bridge. 	In a final phase of each recursive call we compute the SCCs of the graph and for each SCC we remove one
	strong bridge and then recursively call the algorithm on every resulting SCC.
	Before each recursive call, we initialize the lists $L$ to contain the vertices
	that lost an edge during the last phase of the parent recursive call.
	We keep this list in order to restrict the total number of searches for small separable components to $O(m+n)$ since, after initially adding all vertices to the list of the initial call, we only add the endpoints of deleted edges into the lists (which are $O(m)$ many).
	Algorithm~\ref{alg:2ECSC} contains the pseudocode of our algorithm.
	
	The following is a key property that allows us to find small sets that are not strongly connected to the rest of the graph, or that can be disconnected by deleting a single edge, or to conclude that there are no such small sets.
	Every new \EdgeIsolatedOut{1} component that appears in the graph throughout the
	algorithm must have lost an outgoing edge. Respectively, every new
	\EdgeIsolatedIn{1} component that appears must have lost an incoming edge.
	Therefore, we use the list $L$ to keep track of the vertices that have lost an edge
	and for each such vertex $u$ we search for new small \EdgeIsolatedOut{1} or
	\EdgeIsolatedIn{1} components of $u$.
	If no such small components exist in a set of vertices $C$, then we know 
	that either $C$ is a $2$-edge-connected subgraph or
	either recomputing SCCs or the deletion of any
	strong bridge disconnects at least two large components.
	These properties are summarized in the following lemma.

\begin{lemma}
	Let $C$ be a set of vertices of~$G$.
	Every \EdgeIsolatedOut{1} or \EdgeIsolatedIn{1} component \textup{(}of some
		vertex $u\in C$\textup{)} in $G[C]$ that is not such a component in~$G$
		must contain an endpoint of an edge incident to $G[C]$.
	Moreover, if there is no \EdgeIsolatedOut{1} or
		\EdgeIsolatedIn{1} component containing at most $\Delta$ edges
		for any vertex $u \in C$ in $G[C]$, then one of the following holds:
	\begin{enumerate}[\textup{(}a\textup{)}]
		\item $G[C]$ is a $2$-edge-connected subgraph of $G$.\label{sublem:2ec}
		\item There are two sets $A,B\subset C$ with
		$|E(G[A])|, |E(G[B])| > \Delta$ such that $A$ and $B$ are
		in different strongly connected components of $G[C]$.
		\label{sublem:largesccs}
		\item For each strong bridge of $G[C]$ there are two sets
		$A, B \subset C$ with $|E(G[A])|, |E(G[B])| > \Delta$
		that get disconnected by the deletion of the strong bridge.
		\label{sublem:largebridge}
	\end{enumerate}
	\label{lemma:large-set-will-break}
\end{lemma}
\begin{proof}
 		We first show that every \EdgeIsolatedOut{1} component~$\fEdgeIsoOut{1}(u)$
 		of some vertex $u \in C$ that is no \EdgeIsolatedOut{1} component in $G$ must 
 		contain a vertex $x \in \fEdgeIsoOut{1}(u)$ such that there is an edge $(x,y)$ with $y \not\in C$.
 		Assume, by contradiction, that $\fEdgeIsoOut{1}(u)$ exists but there is 
 		no such edge $(x,y)$ in $G$ with $x\in \fEdgeIsoOut{1}(u)$ and $y \notin C$.
 		In this case we have that the very same component $\fEdgeIsoOut{1}(u)$ is a 
 		\EdgeIsolatedOut{1} component of $u$ in $G$.
 		The same argument on the reverse graph shows that every new \EdgeIsolatedIn{1} 
		component (of some vertex $u\in C$) in $G[C]$ must contain an endpoint of an 
 		edge incident to $G[C]$ in $G$.
 		
 		We now turn to the second part of the lemma.
 		If $G[C]$ is strongly connected and does not contain a strong bridge,
 		then $G[C]$ is $2$-edge-connected and thus (\ref{sublem:2ec}) holds.
 		If $G[C]$ is not strongly connected, then it contains (at least)
 		two disjoint sets $A, B \subset C$ such that both $G[A]$ and $G[B]$
 		are strongly connected components of $G[C]$ and 
 		$G[A]$ has no outgoing edge in $G[C]$ (i.e., $G[A]$ is a sink in the DAG 
 		of SCCs of $G[C]$) and $G[B]$ has no incoming edge in $G[C]$ (i.e., $G[B]$ is a 
 		source in the DAG of SCCs of $G[C]$).
 		That is, in $G[C]$ we have that $G[A]$ is or contains a \EdgeIsolatedOut{1} 
 		component (and is a \EdgeIsolatedOut{0} component) of some $u \in C$ 
 		and $G[B]$ contains a \EdgeIsolatedIn{1} component of some $u' \in C$. 
		Both can have the same property in $G$ or contains (resp.\ be) new such 
		components in $G[C]$ compared to $G$.
 		In any case it contradicts the assumptions if one of them has at most $\Delta$ 
 		edges and otherwise statement (\ref{sublem:largesccs}) holds.
 		If $G[C]$ is strongly connected and contains a strong bridge $e^*$, an 
 		analogous argument can be made for two disjoint sets $A, B \subset C$
 		by considering the DAG of SCCs of $G[C] \setminus e^*$. In this case $e^*$
 		is the only incoming edge of $B$ and the only outgoing edge of $A$ in $G[C]$.
 		Thus we have that case~(\ref{sublem:largebridge})
 		holds if the assumptions of the lemma are satisfied.
\proofend
 	\end{proof}
	
	\begin{lemma}\label{lemma:time2ecs}
		Algorithm~$2ECS$ runs in $O(m \sqrt{m})$ time.
	\end{lemma}
	\begin{proof}
		First notice that each time we search for a \EdgeIsolatedOut{1} or a
		\EdgeIsolatedIn{1} component, we are searching for a component with one outgoing
		(resp., incoming) edge containing at most $\sqrt{m}$ edges or with no outgoing
		(resp., incoming) edges and at most $2\sqrt{m}$ edges.
		We can identify if such a component containing a given vertex~$u$ exists
		in time $O(\sqrt{m})$ by using the
		algorithm of Section~\ref{sec:edgeIsolatedOne}.
		We initiate such a search from each vertex that appears in the list $L$ of
		some recursive call of the algorithm.
 		Initially, we place all vertices in the list~$L$.
		Throughout the algorithm we insert into $L$ only vertices that are endpoints of deleted edges.
		Therefore, the number of vertices that are added to the lists~$L$
		throughout the algorithm is $O(m)$.
		Identifying which edges to delete (and thus which vertices to add to $L$)
		can be done in time proportional to the deleted edges and the
		edges in the \EdgeIsolatedOut{1} or \EdgeIsolatedIn{1} component.
		Hence, the total time spent on these searches (and the subsequent operations)
		is $O(m\sqrt{m})$.
		
		Consider now the time spend in each recursive call without the searches
		for  \EdgeIsolatedOut{1} and \EdgeIsolatedIn{1} components.
		Let $G'$ be the graph for which the recursive call is made and let
		$m_{G'} = |E(G')|$.
		In each recursive call the algorithm spends $O(m_{G'})$ time searching
		for strong bridges in $G'$ in lines~\ref{alg-line:Bridges}
		and~\ref{alg-line:strong-bridges-in-sccs} and
		computing SCCs in lines~\ref{alg-line:sccs1} and~\ref{alg-line:sccsdel}.
		Since the subgraphs of different recursive calls at the same recursion depth are
		disjoint, the total time spent at each level of the recursion is $O(m)$.
		We now bound the recursion depth with $O(\sqrt{m})$.
		
		We show that the graph passed to each recursive call has at most $\max\{m_{G'}-\sqrt{m}, 2\sqrt{m}\}$ edges, or $G'$
		is a $2$-edge-connected subgraph
		and thus the recursion stops. This implies a recursion depth of $O(\sqrt{m})$
		as follows.
		If the graph passed to a recursive call has at most $2\sqrt{m}$ edges, then also
		the number of vertices of this graph is at most $2\sqrt{m}$.
		Therefore, even if the algorithm only removes one strong bridge from every
		strongly connected component in each recursive call,
		the total recursion depth is at most $O(\sqrt{m})$.
		On the other hand, the number of times that the graph passed to
		a recursive call has at least $\sqrt{m}$ fewer edges than $G'$ is at most $\sqrt{m}$.
		Overall, this implies that the recursion depth is bounded by $O(\sqrt{m})$.
		
		It remains to show the claimed bound on the size of the graph passed to
		a recursive call in line~\ref{alg-line:recursion-big-sets}.
		For every \EdgeIsolatedOut{1} or \EdgeIsolatedIn{1} component with at most
		$2\sqrt{m}$ edges that is discovered throughout the algorithm, its incident edges
		are removed and therefore it will be in a separate strongly connected component
		with at most $2\sqrt{m}$ edges.
		Let $C$ be the set of vertices that were not included in any
		\EdgeIsolatedOut{1} or \EdgeIsolatedIn{1} component. By
		Lemma~\ref{lemma:large-set-will-break}
		the subgraph $G'[C]$ either is a $2$-edge-connected subgraph or
						there are two sets $A$ and $B$ with $|E(A)|, |E(B)| > \sqrt{m}$ that will be separated in Line \ref{alg-line:strong-bridges-in-sccs}.
		Thus, every graph passed to the recursive call will have at most $\max\{|E(G')|-\sqrt{m}, 2\sqrt{m}\}$ edges.
		The lemma follows.\proofend
	\end{proof}
	
\begin{lemma}
	Let $G = (V, E)$ be a strongly connected digraph. Algorithm $2ECS(G, V)$ returns 
	the maximal $2$-edge-connected subgraphs of $G$.
\end{lemma}
	\begin{proof}
	First note that by assumption the initial call to the algorithm is on a 
	strongly connected graph and that recursive calls are only made on strongly
	connected subgraphs. Thus whenever the algorithm reports a 
	$2$-edge-connected subgraph in 
	line~\ref{alg-line:output-2ec}, then it is a strongly connected subgraph 
	that does not contain any strong bridges, which is by definition 
	a $2$-edge-connected subgraph.
	Thus it suffices to show that the algorithm reports all the maximal $2$-edge-connected 
	subgraphs. Notice that this also implies that the reported $2$-edge-connected 
	subgraphs are maximal.
	Let $C$ be a maximal $2$-edge-connected subgraph.
	We show that the vertices of $C$ do not get separated by the algorithm, and
	therefore $C$ is reported eventually as a $2$-edge-connected subgraph.
	Since there are two edge-disjoint paths between every pair of vertices in~$C$, 
	any search for either a \EdgeIsolatedOut{1} or a \EdgeIsolatedIn{1} component 
	of a vertex $u$ (lines~\ref{alg-line:search-components1}--\ref{alg-line:search-components2}) 
	either returns a superset of~$C$ or 
	fails to identify such a set containing a subset of the vertices of~$C$.
	Furthermore, notice that any deletion of an edge that does not have both endpoints in $C$ does not affect the fact that $C$ is $2$-edge-connected.
	That is, unless an edge with both endpoints in~$C$ is deleted, 
	no strong bridge appears in~$C$.
	Thus, it remains to show that no edge $(x,y)$ such that $x, y \in C$ is 
	ever deleted throughout the algorithm.
	The edges deleted in line~\ref{alg-line:found-component} of the algorithm are 
	incident to a \EdgeIsolatedOut{1} or a \EdgeIsolatedIn{1} component.
	Since $C$ is always fully inside or fully outside of such a set, 
	no edge from $C$ is deleted.
	The edges deleted in line~\ref{alg-line:strong-bridges-in-sccs} are strong bridges
	and the edges deleted in line~\ref{alg-line:sccsdel} before the recursive calls 
	are between separate strongly connected components.
	Since $C$ is $2$-edge-connected, no edges from $C$ are deleted.
	Finally, notice that at each level of recursion at least one of the 
	strong bridges of each strongly connected component of the graph is deleted 
	and the algorithm is recursively executed in each resulting strongly connected
	component. Thus, finally there will be a recursive call for each
	strongly connected subgraph that does not contain 
	strong bridges, including $C$.\proofend
\end{proof}

Algorithm $2ECS$ can be applied to an arbitrary, i.e., not necessarily strongly
connected, digraph by taking the union of the $2$-edge-connected subgraphs of 
the SCCs of the input graph. We have shown the following theorem.
\begin{theorem}\label{th:2ecs}
	The maximal $2$-edge-connected subgraphs of a digraph can be computed in $O(m^{3/2})$~time.
\end{theorem}

\makeatletter{}\section{Maximal $2$-vertex-connected subgraphs in directed graphs}
\label{sec:2vertex}

In this section we first introduce a procedure for identifying \VertexIsolatedOut{1} components
containing at most $\Delta$ edges in time proportional to $\Delta$. 
The same algorithm applied to the reverse graph identifies \VertexIsolatedIn{1} components.
We then use this subroutine with $\Delta = \sqrt{m}$ to obtain a $O(m^{3/2})$
algorithm for computing the maximal $2$-vertex-connected subgraphs of a given 
directed graph.

\subsection{\VertexIsolatedOut{1} and \VertexIsolatedIn{1} components.}
\label{sec:vertexIsolatedOne}

We begin with the definition of \VertexIsolatedOut{k} and \VertexIsolatedIn{k} 
components of a vertex~$u$. In algorithms for $(k+1)$-vertex-connected subgraphs we 
want to detect when for some subgraph (induced by a vertex set) the number of
vertices with outgoing (resp. incoming)
edges has decreased to at most $k$. This can only happen when some vertex has lost 
adjacent edges. Intuitively, the vertex~$u$ for which we search for a 
\VertexIsolatedOut{k} or a \VertexIsolatedIn{k} component is a candidate
for a vertex that is contained in such a subgraph and has lost edges adjacent
to the subgraph. 
Thus it is sufficient to search for \VertexIsolatedOut{k} and \VertexIsolatedIn{k}
components for which the set of vertices with outgoing resp. incoming edges does not
include the starting vertex~$u$. This is reflected in the definitions below
and implies in particular that a  \VertexIsolatedOut{k} 
(resp.\ \VertexIsolatedIn{k}) component contains 
all vertices that have an edge from (resp.\ to) $u$.

\begin{Definition}
Let $G = (V,E)$ be a digraph and $u \in V$ be a vertex.
A \VertexIsolatedOut{k} component of $u$ is a minimal subgraph $S$ of $G$ that contains $u$ and has at most $k$ vertices $X \subset V(S)$, $u \not \in  X$, with outgoing edges to $G\setminus S$.
\end{Definition}

\begin{Definition}
Let $G = (V,E)$ be a digraph and $u \in V$ be a vertex.
A \VertexIsolatedIn{k} component of $u$ is a minimal subgraph $S$ of $G$ that contains $u$ and has at most $k$ vertices $X \subset V(S)$, $u \not \in  X$, with incoming edges from $G\setminus S$.
\end{Definition}

As in the case of \EdgeIsolatedOut{k} (resp., \EdgeIsolatedIn{k}) components, a vertex $u$ may have more than one \VertexIsolatedOut{k} (resp., \VertexIsolatedIn{k}) component.
Also note that for $k' < k$, every \VertexIsolatedOut{k'} component of $u$ is a \VertexIsolatedOut{k} component of $u$ as well.
For the case when $k=1$, the only vertex $x$ that has outgoing (resp., incoming)
edges from a \VertexIsolatedOut{1} (resp., \VertexIsolatedIn{1}) component~$S$ is either a strong articulation point or a vertex that has outgoing (resp., incoming) edges to vertices that belong to different strongly connected components than $x$.
Moreover, each $2$-vertex-connected subgraph is either completely contained in~$S$
or in $(G \setminus S) \cup \{x\}$.

For a given vertex~$u$ and a parameter $\Delta < m/2$,
we present an algorithm for computing a \VertexIsolatedOut{1} component of~$u$
that runs in time $O(\Delta)$ and has the following guarantees:
\begin{itemize}
\item If there exists a \VertexIsolatedOut{1} component of $u$ with at most $\Delta$
edges, then it returns a \VertexIsolatedOut{1} component of $u$ with at most 
$2 \Delta$ edges.
\item If no \VertexIsolatedOut{1} component with at most $\Delta$ edges exists, it might
either return a \VertexIsolatedOut{1} component of $u$ with at most 
$2 \Delta$ edges or the empty set.
\end{itemize}

As mentioned earlier, our algorithm identifies a \VertexIsolatedOut{1} component of $u$ in time proportional to its size (i.e., its number of edges).
In Section \ref{subsec:2VCSs-algorithm} we will use this algorithm to determine quickly whether there exist \VertexIsolatedOut{1} (resp., \VertexIsolatedIn{1}) components of small size (namely, containing at most a predefined number of edges $\Delta$), or conclude that all \VertexIsolatedOut{1} (resp., \VertexIsolatedIn{1}) components have large size. 
We show that this is sufficient to bound the total running time of our algorithm for computing the $2$-vertex-connected subgraphs.

For the rest of this section, we assume that we are given a starting vertex $u$ that can reach at least $2\Delta+1$ edges.
If this is not the case, then the reachable subgraph from~$u$ defines a valid \VertexIsolatedOut{1} component of $u$ that contains at most $2\Delta$ edges and has no outgoing edges.
The exactly same algorithm executed on the reverse graph computes a \VertexIsolatedIn{1} component of $u$ that contains at most $2\Delta$ edges, or we conclude that there is 
no \VertexIsolatedIn{1} component of $u$ with at most $\Delta$ edges.
Since the algorithm for computing a \VertexIsolatedIn{1} component of $u$ is identical to the algorithm for computing a \VertexIsolatedOut{1} component of $u$ when executed on the reverse graph, we only describe the algorithm for finding \VertexIsolatedOut{1} components. The following lemma provides intuition for the properties that we (implicitly)
exploit in our algorithm.

\begin{lemma}
	\label{lemma:2VISO-paths}
	Let $\fVertexIsoOut{1}(u)$ be a \VertexIsolatedOut{1} component of 
	a vertex~$u$ that has outgoing edges and let $ x\not = u$ be the 
	only vertex that has outgoing edges from $\fVertexIsoOut{1}(u)$.
	It holds that $u$ has a path to every vertex $v\in \fVertexIsoOut{1}(u)$ that is contained entirely within the subgraph $\fVertexIsoOut{1}(u)$.
	Moreover, either there is an edge from $u$ to $x$ or 
	$u$ has two internally vertex-disjoint paths to $x$ in $\fVertexIsoOut{1}(u)$.
\end{lemma}
\begin{proof}
	We begin by showing that $u$ has a path to every vertex $v\in \fVertexIsoOut{1}(u)$ that is contained entirely within the subgraph $\fVertexIsoOut{1}(u)$.
	Assume, for the sake of contradiction, that there is a set of vertices $C$ such that the vertices of $C$ are unreachable from $u$ in $\fVertexIsoOut{1}(u)$.
	Then there is no edge $(w,z)$ where $w\in \fVertexIsoOut{1}(u) \setminus C$ and $z \in C$ and thus the outgoing edges from the vertex~$x$ are the only possible 
	outgoing edges from $\fVertexIsoOut{1}(u) \setminus C$.
	Thus, $\fVertexIsoOut{1}(u) \setminus C$ is a \VertexIsolatedOut{1} component of $u$, which contradicts the minimality of $\fVertexIsoOut{1}(u)$.
	
	We now show that if there is no edge from $u$ to $x$, then $u$ has two 
	internally vertex-disjoint paths to $x$ in $\fVertexIsoOut{1}(u)$.
	First, we note that all simple paths from $u$ to $x$ contain only vertices in $\fVertexIsoOut{1}(u)$ since there is no other vertex $x' \not = x$ such that $x' \in \fVertexIsoOut{1}(u)$ and $x'$ has edges leaving $\fVertexIsoOut{1}(u)$.
	Assume, for the sake of contradiction, that all paths from $u$ to $x$ in $\fVertexIsoOut{1}(u)$ share a common vertex $w$ that is different from 
	both $u$ and $x$.
	Then, $u$ does not have a path to $x$ in $\fVertexIsoOut{1}(u) \setminus \{w\}$.
	Let $C$ be the set of vertices that become unreachable from $u$ in $\fVertexIsoOut{1}(u) \setminus \{w\}$. (Notice that $|C| \geq 1$ since $x \in C$.) 
	Clearly, there is no edge $(w',z')$ such that $w' \in \fVertexIsoOut{1}(u) \setminus (C \cup \{w\})$ and $z' \in C$, 
	since otherwise $z'$ would be reachable from $u$ in
	$\fVertexIsoOut{1}(u) \setminus \{w\}$. Hence, the only vertex that has 
	edges leaving $\fVertexIsoOut{1}(u) \setminus C$ is $w$.
	Thus, $\fVertexIsoOut{1}(u) \setminus C$ is a \VertexIsolatedOut{1} component 
	of $u$, which contradicts the minimality of $\fVertexIsoOut{1}(u)$. The lemma follows.\proofend
\end{proof}

Our algorithm for identifying \VertexIsolatedOut{1} components 
begins with a DFS traversal $F_1$ from $u$. As for \EdgeIsolatedOut{1} components,
the idea is that the first DFS traversal ``uses and blocks'' the only
vertex that has edges out of a \VertexIsolatedOut{1} component
if such a component of size at most $\Delta$ exists,
and then a second traversal explores exactly 
the \VertexIsolatedOut{1} component.
In the DFS traversal $F_1$ we
charge to a visited vertex its outgoing edges that were traversed.
We stop $F_1$ when the number of the traversed edges reaches $2\Delta+1$.
Let $T$ be the DFS tree constructed by the DFS traversal.
We define the weight of a vertex $v$, denoted by $w(v)$, to be the total number of 
edges charged to the descendants of $v$ in $T$ (including~$v$). 

Assume that $u$ has a \VertexIsolatedOut{1} component $C$ containing at most 
$\Delta$ edges and exactly one vertex~$x$ with outgoing edges to $V\setminus C$.
It is easy to see that $F_1$ is guaranteed to traverse at least $\Delta+1$ edges
outside of $C$ (since it visits at least $2\Delta+1$ edges and $|E(C)|\leq \Delta$),
and therefore, since $x$ is the only vertex with outgoing edges from $C$, we have $w(x)\geq \Delta+1$.
Moreover, for any vertex $v\not = u$ whose DFS subtree explores only vertices inside $C$, we have $w(v)<\Delta$. 
The following two lemmata show that the vertices with weight more than $\Delta$
form a path from $u$ to a vertex outside of $C$ that we can then use to block
the only vertex with outgoing edges for the second traversal starting from $u$.

\begin{lemma}
	Let $\fVertexIsoOut{1}(u)$ be a \VertexIsolatedOut{1} component of $u$ such that $|E(\fVertexIsoOut{1}(u))| \leq \Delta$, let $x$ be the only vertex that has edges leaving $\fVertexIsoOut{1}(u)$, and let $T$ be a DFS tree generated by a DFS traversal from $u$ that visits $2\Delta+1$ edges.
	Then, for each $v \in T[u,x]$ it holds that $w(v) \geq  \Delta+1$ and for each $v \in \fVertexIsoOut{1}(u) \setminus T[u,x]$ it holds that $w(v) \leq \Delta$.
	\label{lemma:2VISO-blocked-edges}
\end{lemma}
\begin{proof}
	Since $\fVertexIsoOut{1}(u)$ contains at most $\Delta$ edges, the only way a DFS 
	traversal can visit $2\Delta+1$ edges is by visiting at least $\Delta+1$ edges 
	outside of $\fVertexIsoOut{1}(u)$. By the fact that $x$ is the only vertex that has 
	edges leaving $\fVertexIsoOut{1}(u)$, it follows that 
	$w(x)\geq \Delta+1$, and therefore, for each $v \in T[u,x]$ it holds 
	that $w(v)\geq \Delta+1$.
	Note that the DFS reaches each vertex and in particular $x$ only once (i.e.,
	each vertex of $T$ except $u$ has exactly one incoming edge in $T$), 
	and that any traversal from $u$ that does not visit vertices 
	$v \notin \fVertexIsoOut{1}(u)$ cannot be charged more than $\Delta$ edges.
	None of the subtraversals from vertices $v \in \fVertexIsoOut{1}(u)
	\setminus T[u,x]$ can visit vertices outside of $\fVertexIsoOut{1}(u)$ since 
	either $v$ is visited after $x$ or the subtraversal can not reach $x$.
	In both cases the subtraversal from~$v$ can not use the outgoing edges 
	of $x$ to visit more than $\Delta$ edges. Thus, for each vertex $v \in
	\fVertexIsoOut{1}(u) \setminus T[u,x]$, it holds that $w(v) \leq \Delta$.\proofend
\end{proof}

After the traversal $F_1$, we say that a vertex $v$ is \emph{blocked} if 
$w(v)\geq \Delta+1$. Next, we start a second traversal $F_2$ from $u$ 
(not necessarily a depth-first search) as follows. The traversal $F_2$ can only 
visit the vertex~$u$ and vertices that are not blocked. We say that the 
traversal \emph{reaches} a vertex $v$ whenever it traverses an edge incoming to $v$; 
thus $F_2$ can reach blocked vertices but not visit them and all vertices 
that are visited are also reached by $F_2$.
Whenever $F_2$ reaches a blocked vertex $v$, we unblock all blocked vertices 
on $T[u,v] \setminus v$. (Notice that $v$ itself is not unblocked.)
Assuming that there exists a \VertexIsolatedOut{1} component of $u$ with at most
$\Delta$ edges for which $x \ne u$ is the only vertex with outgoing edges. Then this 
second traversal $F_2$ has two main properties: (\textit{i}) it never unblocks~$x$,
and (\textit{ii}) it reaches all vertices in $\fVertexIsoOut{1}(u)$.
Since we are interested only in computing a \VertexIsolatedOut{1} component of $u$
containing at most $\Delta$ edges (recall that we assumed in the beginning 
that $u$ can reach at least $2\Delta+1$ edges), we terminate $F_2$ whenever it
visits $\Delta+1$ edges. 
If the traversal $F_2$ visits $\Delta+1$ edges, we conclude that there is no \VertexIsolatedOut{1} component of $u$ containing at most $\Delta$ edges.
Before proving the above claim, we first show the following supporting lemma, which says that the blocked vertices form a path in the DFS tree; we call this path the 
\emph{heavy path} of $F_1$.

\begin{lemma}
	\label{lemma:vertex-blocked-path}
	Let $F$ be a DFS traversal that visits $2\Delta+1$ edges and let $T$ be its 
	DFS tree. The vertices~$v$ with $w(v) \geq \Delta+1$ form a path in $T$.
\end{lemma}
\begin{proof}
  Assume, by contradiction, that  the vertices~$v$ with $w(v) \geq \Delta+1$ do 
  not form a path on $T$. That means, there are two vertices $x$ and $y$
  with $w(x),w(y)\geq \Delta+1$ that do not have an ancestor-descendant relation 
  in $T$, i.e., $T(x) \cap T(y) = \emptyset$.
		This is a contradiction to the fact that $F$ visits only $2\Delta+1$ edges.
	Therefore, the vertices~$v$ with $w(v) \geq \Delta+1$ form a path in $T$.\proofend
\end{proof}

	\begin{procedure}[t!]
		\DontPrintSemicolon
		\KwIn{Digraph $G=(V,E)$, a vertex $u$, and an integer $\Delta$}
		\KwOut{Either a \VertexIsolatedOut{1} component of $u$ with at most $2\Delta$ edges or $\emptyset$; if $\emptyset$ is returned, then no \VertexIsolatedOut{1} component of $u$ with at most $\Delta$ edges exists}
		\BlankLine
		
		Execute DFS $F_1$ from~$u$ for up to $2\Delta+1$ edges\;
		Let $S_1$ be the vertices reached by $F_1$
		
		\If{ $F_1$ cannot reach $2\Delta+1$ edges}{\Return $G[S_1]$ as \EdgeIsolatedOut{1} component of $u$}
		\Else{
			
			Block the vertices on the heavy path of $F_1$\;
			
			Execute a DFS $F_2$ from~$u$ for up to $\Delta+1$ edges and
			whenever a blocked vertex~$v$ is reached: unblock vertices 
			from $u$ to 
			the predecessor of $v$ in $F_1$ and continue the DFS without~$v$\;
			Let $S_2$ be the vertices reached by $F_2$ (including reached but not unblocked vertices)
			
			\If{$F_2$ cannot reach $\Delta+1$ edges} {
				\Return $G[S_2]$ as \VertexIsolatedOut{1} component of~$u$
			}
			\Else{\Return $\emptyset$}
		}
		\caption{1VertexOut($G$, $u$, $\Delta$)}
		\label{proc:1VertexOut}
	\end{procedure}

\begin{lemma}
	Let $G$ be a graph where the vertices $v$ with $w(v) \geq \Delta+1$ are blocked after the DFS traversal $F_1$.
	If there exists a \VertexIsolatedOut{1} component of $u$ containing at most $\Delta$ edges, then $F_2$ traverses at most $\Delta$ edges.
	Moreover, if $F_2$ traverses at most $\Delta$ edges, the subgraph induced by
	the vertices reached by $F_2$ \textup{(}including a reached but not unblocked vertex\textup{)} defines a \VertexIsolatedOut{1} component of $u$ that contains at most $\Delta + 1$ vertices and at most $2\Delta$ edges.
\end{lemma}
\begin{proof}
		Let us first assume that there exists a \VertexIsolatedOut{1} component $\fVertexIsoOut{1}(u)$ of $u$ that contains at most $\Delta$ edges and that all edges leaving $\fVertexIsoOut{1}(u)$ share a common source $x$.
	By Lemma \ref{lemma:2VISO-blocked-edges}, $x$ is blocked. 
	The traversal $F_2$ cannot visit more than $\Delta$ edges, since $u$ cannot visit vertices $v \notin \fVertexIsoOut{1}(u)$ avoiding $x$, and hence, $F_2$ cannot unblock $x$.
												
	Now we show the opposite direction.
	Assume that $F_2$ visits at most $\Delta$ edges and thus reaches at most $\Delta+1$
	vertices. We will show that there exists a \VertexIsolatedOut{1} component $\fVertexIsoOut{1}(u)$ of $u$ that is induced by the vertices reached by $F_2$
	and has at most $2 \Delta$ edges.
	If $F_2$ unblocks the whole path $P_{blocked}$, then it would visit at least $2\Delta+1$ edges, since $F_1$ did so.
	Hence, there is at least one vertex that remains blocked after the traversal of $F_2$; let $v^*$ be such a vertex.
	Let $C$ be the set of vertices that were reached by $F_2$.
	Then, $C$ has at most one blocked vertex, which is $v^*$, since whenever two vertices of the path $P_{blocked}$ are reached, reaching the vertex further 
	away from $u$ on $P_{blocked}$ unblocks all the blocked vertices on the tree
	path from~$u$.
	Notice that all edges leaving $C$ are from~$v^*$.
	Moreover, $v^*$ might have at most $\Delta$ edges to vertices in $C$ that were not traversed.
	Thus the subgraph induced by~$C$ contains $u$, has only one vertex $v^* \ne u$
	with edges out of the subgraph, and contains at most $2\Delta$ edges. Notice 
	that all vertices in $C$ were reached by $F_2$.
	To show that $C$ induces a \VertexIsolatedOut{1} component of $u$ it 
	remains to show that there is no proper subset $C'$ of $C$ that contains~$u$ 
	and has at most one vertex $x' \ne u$ with edges out of $C'$.
	Assume by contradiction that there exists such a vertex set $C'$.
	By Lemma~\ref{lemma:2VISO-blocked-edges} the traversal $F_1$ would have blocked
	$x'$, and there is no other outgoing edge from a vertex of $C'$ to a vertex 
	in $V \setminus C$. 
	Therefore, $F_2$ cannot visit vertices outside of $C$ since it cannot unblock 
	$x'$. This is a contradiction to the fact that $F_2$ visits all the vertices of $C$.\proofend
\end{proof}

	After the execution of the traversal $F_2$, we can either return a \VertexIsolatedOut{1} component of $u$ with at most $2\Delta$ edges or decide that all \VertexIsolatedOut{1} components of~$u$ contain more than $\Delta$ edges, as shown in Lemma~\ref{lemma:determin-1-edge-out-component}.
	The pseudocode of our algorithm is illustrated in Procedure~\ref{proc:1VertexOut}.
	The following lemma summarizes the result of this section.

\begin{lemma}
	Procedure~\ref{proc:1VertexOut} computes in $O(\Delta)$ time a
	\VertexIsolatedOut{1} component of a vertex $u$ containing at most $2\Delta$
	edges or decides that there is no \VertexIsolatedOut{1} component of $u$ 
	containing at most $\Delta$ edges.
\end{lemma}

\subsection{Computing the $2$-vertex-connected subgraphs.}
\label{subsec:2VCSs-algorithm}

In this section we present an $O(m \sqrt{m})$ time algorithm for computing the $2$-vertex-connected subgraphs of a directed graph.
We begin with a simple algorithm and then show how we can improve its running time.
Recall that the $2$-vertex-connected subgraphs of a graph are subgraphs that do not contain any strong articulation points, that is, they cannot get disconnected by the deletion of any single vertex.
In contrast to $2$-edge-connected subgraphs, the $2$-vertex-connected subgraph
do not define a partition of the vertices of the input graph.
More specifically, any two $2$-vertex-connected subgraphs might share up to one 
common vertex.
This introduces an additional challenge since the existence of a strong articulation point $x$ that disconnects a vertex set $S$
guarantees that no vertex of $S$ appears in the same $2$-vertex-connected 
subgraph as a vertex of $V\setminus (S \cup x)$ but does not provide information 
on whether $x$ itself appears in a $2$-vertex-connected subgraph with vertices 
from $S$ or $V\setminus (S \cup x)$.

A simple algorithm for computing the $2$-vertex-connected subgraphs of a directed graph works as follows. 
Assume the input graph is strongly connected (or consider each SCC separately).
We repeatedly find a strong articulation point $x$ that disconnects the graph into two sets of vertices $S$ and $V \setminus (S \cup x)$, i.e., there is no pair of vertices $u$ and $v$ that are strongly connected in $G\setminus x$ such that $u\in S$ and $v\in V\setminus (S \cup x)$.
We recursively execute the same algorithm on the strongly connected components of the subgraphs $G[S\cup x]$ and $G[V\setminus S]$ that contain at least three vertices.
If a recursive call fails to identify a strong articulation point in a strongly connected subgraph, then it reports the subgraph as $2$-vertex-connected.
The correctness and the running time of this simple algorithm can easily 
be verified along the following lines.
First, since at each recursive call we identify a strong articulation point that
separates two (non-empty) sets of vertices, we know that the $2$-vertex-connected 
subgraphs of these two sets are disjoint apart from possibly the articulation point
itself. Moreover, we restrict the recursive calls to the strongly connected components of 
the resulting subgraphs since every $2$-vertex-connected subgraph is also strongly
connected. Therefore, all the $2$-vertex-connected subgraphs are preserved at 
each recursive call. The algorithm reports a $2$-vertex-connected subgraph 
once it recurses on a subgraph that does not contain a strong articulation point, 
which is correct by definition. Second, we bound the running time.
The maximum recursion depth is $O(n)$ 
since every recursive call is executed on a 
graph that contains at least one vertex less than the parent call.
Although at each recursive call the strong articulation point is included in 
both sets that it separates, the set of edges is partitioned between the two subgraphs. 
Therefore, at each level of recursion the total number of edges in all instances is 
at most $m$, and the total time to compute a strong articulation point and the 
strongly connected components at the end of each recursive call is $O(m)$, which
leads to an overall running time of $O(mn)$.

The high-level idea of our algorithm for computing the $2$-vertex-connected 
subgraphs is similar to the algorithm of Section 
\ref{sec:2-edge-connected-strong-components} for computing the $2$-edge-connected
subgraphs. We additionally define the following operation to construct 
the subgraphs on which the algorithm recurses.
Let $G$ be a digraph, $x$ a vertex, and $N$ a subset of neighbors of $x$.
The operation $split(x,N)$ is executed as follows.
First, we create an additional vertex $x'$ in $G$, that serves as a copy of $x$.
Second, for every edge $(x,y)$ with $y \in N$ we remove $(x,y)$ from~$G$ and 
add the edge $(x',y)$.
Analogously, for every edge $(y,x)$ with $y\in N$ we remove $(y,x)$ from~$G$ and add the edge $(y,x')$.
This operation can be implemented to take time proportional to the number of 
neighbors of the vertices in $N$ by traversing their edges and 
change every edge that is incident to $x$ to be incident to~$x'$.

\begin{lemma}
  A $split$ operation preserves the number of edges in the graph.
	The maximum number of auxiliary vertices after any sequence of $split$ operations is $2m-n$.
	\label{lemma:edges-vertices-split}
\end{lemma}
\begin{proof}
	By definition, no edges are added or deleted when performing the $split$ operation.
	Since every edge has two endpoints, in the worst case each vertex is adjacent 
	to only one edge.
	Notice that the original $n$ vertices always exist in the graph.
	Therefore, the total number of auxiliary vertices cannot exceed $2m-n$.\proofend
\end{proof}

\begin{lemma} 
	Let $G$ be a directed graph, $x$ a strong articulation point, and let 
	the sets $N_1$, $N_2$ be a partition of the vertices adjacent to $x$ such 
	that all paths from vertices in $N_1$ to vertices in $N_2$ go through $x$.
	There is a one-to-one correspondence between the $2$-vertex-connected 
	subgraphs in $G$ and in the graph resulting from $G$ through
	the execution of either $split(x,N_1)$ or $split(x,N_2)$.
	\label{lemma:sccs-after-split}
\end{lemma}
\begin{proof}
	W.l.o.g., we assume that the $split$ operation is $split(x,N_1)$.
	Let $C$ be a $2$-vertex-connected subgraph before the execution of the 
	$split$ operation. If the $split$ operation is not executed on a vertex of
	$C$, then $C$ remains a $2$-vertex-connected subgraph by the definition of $x$.
	Now assume that the split operation is executed on a vertex $x\in C$.
	Then all neighbors of $x$ that are in $C$ are strongly connected in $G\setminus x$, 
	and therefore they are either all included in $N_1$ or none of them is.
	Thus, all the edges between the vertices of $C$ are preserved.
	
	Now we prove the opposite direction. Let $C$ be a $2$-vertex-connected 
	subgraph after the execution of the $split$ operation.
	Then, either all edges between the vertices of $C$ existed before the $split$ 
	operation, or there is an auxiliary vertex $x' \in C$ such that all edges between 
	vertices of $C\setminus x'$ existed before the operation and all edges between 
	vertices of $C \setminus x'$ and $x'$ were between $C \setminus x'$ and a vertex 
	$x$ before the $split$ operation (where $x$ is the vertex on which the 
	$split$ operation was executed). That is, no additional paths among the 
	vertices of $C$ were introduced through the split operation and thus in both 
	cases $C$ was a $2$-vertex-connected subgraph before the $split$ operation.\proofend
\end{proof}

We are ready to describe our algorithm for computing the $2$-vertex-connected 
subgraphs of a directed graph $G$. We build on the simple recursive algorithm 
that is described at the beginning of this section.
To distinguish the input graph from the graphs in the recursive calls, 
we refer to the original input graph as $G_0=(V_0,E_0)$.
We use the terms small components and large components to refer to subgraphs that contains at most and more than $\sqrt{m_0}$ edges, respectively, where $m_0 = |E_0|$. (We allow small components to contain up to $2\sqrt{m_0}$ edges.)
Our algorithm begins by identifying all the small \VertexIsolatedOut{1} and \VertexIsolatedIn{1} components of any vertex in $G_0$, using the algorithm from Section~\ref{sec:vertexIsolatedOne}.
Throughout the algorithm we maintain a list $L$ of the vertices for which we then
start a search for a small \VertexIsolatedOut{1} or \VertexIsolatedIn{1} component.
We show that it is sufficient to search from the vertices that are inserted
into $L$ throughout the algorithm in order to find \emph{all} the small
\VertexIsolatedOut{1} and \VertexIsolatedIn{1} components of \emph{all} the 
vertices in the graph.
In the initial call to the algorithm we set $L = V_0$ (i.e., this is not done for every recursive call).
At each recursive call the algorithm first tests whether the given strongly connected
graph contains less than $3$ vertices; in this case the empty set is returned as 
every $2$-vertex-connected subgraph has to contain at least $3$ vertices. Then it is
tested whether the given strongly connected graph is also $2$-vertex-connected, which
is the case (by definition) if it does not contain a strong articulation point; in 
this case the algorithm outputs the graph as a $2$-vertex-connected subgraph.
Then, while $L$ is not empty, we extract a vertex~$u$ from~$L$ and search for a 
small \VertexIsolatedOut{1} or a small \VertexIsolatedIn{1} component of $u$
(containing at most $2\sqrt{m}$ edges). 

After identifying a \VertexIsolatedOut{1} component $\fVertexIsoOut{1}(u)$ 
whose outgoing edges originate from a common vertex~$x$, the algorithm executes
$split(x,N)$, where $N$ are the neighbors of $x$ in $\fVertexIsoOut{1}(u)$.
Furthermore, for every edge $e=(w,z)$ incident to $\fVertexIsoOut{1}(u)$ that is
not adjacent to $x$ (i.e., the incoming edges of $\fVertexIsoOut{1}(u)$), we 
insert both $w$ and $z$ into $L$ and remove $e$ from the graph.
We treat every identified small \VertexIsolatedIn{1} component in an analogous way.

If, on the other hand, the graph is not $2$-vertex-connected but
we cannot find a new small  \VertexIsolatedOut{1} or 
\VertexIsolatedIn{1} component, we conclude that either there are at least two 
large sets of vertices that are in different strongly connected components or 
for every strong articulation point there exist two large sets of vertices that
get disconnected by the removal of the strong articulation point.
To exploit that, we perform the following steps in the final phase of each 
recursive call: First we compute the strongly 
connected components $C_1,C_2, \dots, C_c$ of the graph that results from 
the $split$ operations. Each of the SCCs that does not contain a strong 
articulation point and contains at least $3$ vertices is $2$-vertex-connected
and added to the set of $2$-vertex-connected subgraphs. For each of
the other SCCs $C_i$ we first execute split $(v,N_{C'})$ on some strong 
articulation point $v$, where $N_{C'}$ are the neighbors of $v$ that are
contained in a singe arbitrary strongly connected component $C'$ in 
$G[C_i]\setminus v$, and then recursively call the algorithm on each strongly
connected component of the resulting graph. 
Before each recursive call we initialize the lists $L$ to contain the vertices 
that lost an edge during the last phase of the parent recursive call.
After initially adding all vertices to the list of the initial call, 
we only add the endpoints of deleted edges to the lists, thus the 
total number of searches for small  \VertexIsolatedOut{1} and  
\VertexIsolatedIn{1} components is bounded by $O(m+n)$.
Algorithm~\ref{alg:2VCS} contains the pseudocode of our algorithm.
This formulation of the algorithm has the advantage that we can bound the 
size of the subgraphs passed to the recursive calls: either the subgraph 
is a small  \VertexIsolatedOut{1} or \VertexIsolatedIn{1} component
and thus contains at most $2 \sqrt{m}$ edges or two large sets are separated 
and therefore the number of edges for each subgraph at the subsequent level
of recursion is reduced by at least $\sqrt{m}$, which can happen at most $\sqrt{m}$
times.

Similarly to Algorithm~\ref{alg:2ECSC} from 
Section~\ref{sec:2-edge-connected-strong-components}, we now show the key property 
that allows us to either find small sets that can be separated by a single vertex
deletion or conclude that there are at least two large vertex sets that are either 
not strongly connected to each other or become disconnected by the deletion of a 
single vertex. 
Every new \VertexIsolatedOut{1} component that appears in the graph throughout 
the algorithm must contain a vertex that has lost all its outgoing edges that 
led to vertices not in the component as otherwise the component would have 
been a \VertexIsolatedOut{1} component before.
Note that this vertex that has lost outgoing edges cannot be equal to the only 
vertex that \emph{still} has outgoing edges from the \VertexIsolatedOut{1} component.
Analogously, every new \VertexIsolatedIn{1} component that appears must have 
lost an incoming edge to a vertex other than the separating
vertex of the \VertexIsolatedIn{1} component.
Therefore, we use the list $L$ to keep track of the vertices that have lost an edge and for each such vertex $u$ we search for new small \VertexIsolatedOut{1} or \VertexIsolatedIn{1} components of $u$.
If no such small components exist in a set of vertices $C$, then we know that either (i) $C$ is a $2$-vertex-connected subgraph or (ii) we are guaranteed that either two large sets of vertices are in separate strongly connected components of the graph, or that every strong articulation point separates two large sets of vertices.
This property is summarized in the following lemma.

\begin{lemma}
Let $C$ be a set of vertices in $G$.
		Each \VertexIsolatedOut{1} component \textup{(}of some 
 		vertex $u\in C$\textup{)} in $G[C]$ for which $x$ is the only vertex that has outgoing edges to $V\setminus C$ and that is not a \VertexIsolatedOut{1} component in~$G$ must contain an endpoint $z$ of an edge incident 
 		to $G[C]$, such that $z\not = x$.
		Moreover, if there is no \VertexIsolatedOut{1} or 
 		\VertexIsolatedIn{1} component containing at most $\Delta$ edges
 		for any vertex $u \in C$ in $G[C]$, then one of the following holds.
 		\begin{enumerate}[\textup{(}a\textup{)}]
 		\item $G[C]$ is a $2$-vertex-connected subgraph.\label{case:2vcs}
 		\item There are two sets $A,B\subset C$ with $|E(G[A])|, |E(G[B])| > \Delta$ that are disjoint strongly connected components.\label{case:2v:sc}
 		\item For every strong articulation point $x$, there are two sets $A,B\subset C$ with $|E(G[A])|, |E(G[B])| > \Delta$ that are separated in $G[C]\setminus x$.\label{case:art}
 		\end{enumerate}
	\label{lemma:large-set-will-break-vertex}
\end{lemma}
\begin{proof}
 		We first show that every \VertexIsolatedOut{1} component~$\fVertexIsoOut{1}(u)$
 		of some vertex $u \in C$ for which $x$ is the only vertex that has 
 		outgoing edges to $V\setminus C$ and that is 
 		no \VertexIsolatedOut{1} component in $G$ must 
 		contain a vertex $w \in \fVertexIsoOut{1}(u) \setminus x$ such 
 		that there is an edge $(w,y) \in G$ with $y \not\in C$.
 		Assume, by contradiction, that $\fVertexIsoOut{1}(u)$ exists but there is 
 		no such edge $(w,y)$ in $G$ with $w\in \fVertexIsoOut{1}(u) \setminus x$ 
 		and $y \notin C$.
 		In this case, the very same component $\fVertexIsoOut{1}(u)$ is a 
 		\VertexIsolatedOut{1} component of $u$ in $G$, since $x$ is the only 
 		vertex having outgoing edges to $V\setminus C$.
 		The same argument on the reverse graph shows that every \VertexIsolatedIn{1}
 		component (of some vertex $u\in C$) in $G[C]$ must contain an endpoint of an
 		edge incident to $G[C]$.
	
	Now we turn to the second part of the lemma.
	If $G[C]$ is strongly connected and does not contain an articulation point,
	then $G[C]$ is $2$-vertex-connected, i.e., case~\eqref{case:2vcs} holds.
	If $G[C]$ is not strongly connected,  then it contains (at least)
 		two disjoint sets $A, B \subset C$ such that both $G[A]$ and $G[B]$
 		are strongly connected components of $G[C]$ and 
 		$G[A]$ has no outgoing edge in $G[C]$ (i.e., $G[A]$ is a sink in the DAG 
 		of SCCs of $G[C]$) and $G[B]$ has no incoming edge in $G[C]$ (i.e., $G[B]$ is a 
 		source in the DAG of SCCs of $G[C]$).
 		That is, in $G[C]$ we have that $G[A]$ is or contains a \VertexIsolatedOut{1} component of some $u \in C$
 		and $G[B]$ is or contains a \VertexIsolatedIn{1} component of some $u' \in C$. 
		Both can have the same property in $G$ or contains (resp. be) new such components in $G[C]$ compared to $G$.
 		In any case it contradicts the assumptions if one of them has at most $\Delta$ 
 		edges and otherwise case~\eqref{case:2v:sc} holds.
If $G[C]$ is strongly connected and contains an articulation point $v^*$, an 
 		analogous argument can be made for two disjoint sets $A, B \subset C$
 		by considering the DAG of SCCs of $G[C] \setminus v^*$. In this case $v^*$
 		is the only vertex with incoming edges of $B$ and the only vertex with 
 		outgoing edges of $A$ in $G[C]$. Thus in this case \eqref{case:art} is 
 		satisfied if the assumptions of the lemma hold.\proofend
\end{proof}

\begin{algorithm}[t!]
		\DontPrintSemicolon
		\KwIn{		A strongly connected digraph $G=(V,E)$ and
		a list of vertices $L$ (initially $L = V$)}
		\KwOut{The $2$-vertex-connected subgraphs of $G$}
		\BlankLine
		Let $m_0$ be the number of edges of the initial graph\;
		\lIf(\tcp*[h]{removing degenerate subgraphs}){$\vert V \rvert \le 2$}{\Return $\emptyset\;$}
		\If{$G$ has no strong articulation point\label{alg-line:SAPs}}{
			\Return $\{G\}$ as $2$-vertex-connected subgraph
		\label{alg-line:output-2vc}
		}		
		\While{$L \not = \emptyset$ and $G$ has more than $2 \sqrt{m_0}$ edges}
		{
			Extract a vertex $u$ from $L$
			
			$S \gets $\ref{proc:1VertexOut}($G$, $u$, $\sqrt{m_0}$) \label{alg-line:search-components1v}\;
			$S^R \gets $\ref{proc:1VertexOut}($G^R$, $u$, $\sqrt{m_0}$)
			\label{alg-line:search-components2v}\;
			
			Pick non-empty set of $S$ and $S^R$ if it exists\;
			Let $x$ be the common vertex in $S$ resp.~$S^R$ of all outgoing resp.~incoming edges (if it exists) and let $N$ be the neighbors of $x$ inside the set\;
			Execute $split(x, N)$ (if $x$ exists)\;
			Delete all edges incident to the selected set that are not adjacent to $x$ and add their endpoints to $L$
		}
		
		Compute strongly connected components $C_1, \ldots, C_c$ of $G$\label{alg-line:2vccs}\;
		$U \gets \emptyset$\;
		\ForEach{$C_i, 1 \leq i \leq c$}{
		  \If{$C_i$ contains a strong articulation point~$v$}{
			execute $split(v,N_{C'})$, where $N_{C'}$ are the edges between $v$ and the vertices of an arbitrary strongly connected component $C'$ of $C_i\setminus v$. \label{alg-line:sccs-v}\;
			\ForEach{SCC $C$ of $C_i$}{
				Initialize $L'$ with the vertices of $C$ that are endpoints of 
				newly deleted edges\;
				$U \gets U \cup 2VCS(C,L')$\label{alg-line:recursion-big-sets-v}\;
			}
		  }\Else{
			\If(\tcp*[h]{$C_i$ is $2$-vertex-connected}){$|V(C_i)| \ge 3$}{
			  $U \gets U \cup \{C_i\}$\;
			}
		  }
		  
		}
		\Return $U$
		\caption{\textsf{$2VCS(G, L)$ }}
		\label{alg:2VCS}
	\end{algorithm}

\begin{lemma}
	Algorithm \textsf{$2VCS$} is correct.
\end{lemma}
\begin{proof}
First note that by assumption the initial call to the algorithm is on a 
	strongly connected graph and that recursive calls are only made on strongly
	connected subgraphs. Thus whenever Algorithm~\textsf{$2VCS$} reports a 
	$2$-vertex-connected subgraph, then this is a $2$-vertex-connected subgraph, 
	since it is strongly connected, does not have any strong articulation points,
	and contains at least $3$ vertices.
	It suffices to show that \textsf{$2VCS$} reports all the maximal 
	$2$-vertex-connected subgraphs. Notice that this also implies that the 
	reported $2$-vertex-connected subgraphs are maximal.
	Let $C$ be a maximal $2$-vertex-connected subgraph.
	We show that $C$ does not get disconnected by the algorithm, since this will ensure that the algorithm eventually will recurse on $C$ and report it as a $2$-vertex-connected subgraph.
	Since there is no vertex whose deletion separates any pair of vertices in $C$, any search for either a \VertexIsolatedOut{1} or a \VertexIsolatedIn{1} component (of some vertex $u$), either returns a superset of $C$, or it fails to identify such a set containing a subset of the vertices of $C$.
	Furthermore, note that any deletion of an edge that does not have both endpoints in $C$ does not affect the fact that $C$ is $2$-vertex-connected.
	That is, unless an edge with both endpoints in $C$ is deleted, no strong articulation points appear in $C$.
	Thus, it is left to show that no edge $(x,y)$ such that $x, y\in C$ is ever deleted throughout the algorithm.
	The edges that are deleted are either edges between strongly connected components,
	or between two sets of vertices $A,B$ that get disconnected by a strong
	articulation point, or edges incident to a \VertexIsolatedOut{1} or a
	\VertexIsolatedIn{1} component found during the course of the algorithm.
	Since $C$ is always fully included in such a component, no edge of $C$ is deleted.
	Finally, notice that in each recursive call, unless the graph that is passed 
	to the recursion is $2$-vertex-connected, at least one strong articulation 
	point that separates at least one pair of vertices is computed and the 
	algorithm recurses on each strongly connected component (possibly containing 
	a copy of the strong articulation point) after its removal.
	Thus, the algorithm makes progress in each iteration and 
	at some point there will be a recursive call for each strongly connected 
	subgraph that does not contain strong articulation points, including $C$.\proofend
\end{proof}

\begin{lemma}
	Algorithm \textsf{$2VCS$} runs in $O(m \sqrt{m})$ time on a graph with $m$ edges.
\end{lemma}
\begin{proof}
	Let $G_0 = (V_0,E_0)$ be the input graph for the initial call to the algorithm.
	Let $n_0 = |V_0|$ and $m_0 = |E_0|$.
	First, notice that each time we search for a \VertexIsolatedOut{1} 
	(or a \VertexIsolatedIn{1} component by searching on the reverse graph), 
	we are searching either for a component where all outgoing edges have a common 
	source or for a component with no outgoing edges; in both cases we search for 
	a component with at most $2\sqrt{m_0}$ edges. 
	We can identify if such components exist in time $O(\sqrt{m_0})$ by using the
	algorithm of Section \ref{sec:vertexIsolatedOne}.
	We start a search from every vertex that is added to the list $L$ in some 
	recursive call. Notice that initially we add all vertices to $L$, and 
	throughout the course of the algorithm we insert the two endpoints of every
	deleted edge into the corresponding list~$L$. By Lemma~\ref{lemma:edges-vertices-split} the number of edges does not increase by the $split$ operations.
	Therefore, the total time spent on these calls is $O((m_0+n_0)\sqrt{m_0}) = O(m_0 \sqrt{m_0})$. 
	For every \VertexIsolatedOut{1} or \VertexIsolatedIn{1} component~$S$ (with 
	at most $2\sqrt{m_0}$ edges) that is discovered throughout the algorithm, the
	component has either no outgoing (resp., incoming) edges or we execute 
	the $split$ operation on the only vertex $x$ that has outgoing (resp., incoming)
	edges. 	We can execute the operation $split$ in time proportional to the edges 
	incident to the neighbors of $x$ in $S$, and we can charge this time, as well
	as the time for identifying the edges to delete, to the 
	process of identifying the set $S$ (that covers for the edges in $G'[S]$),
	and to the edges deleted from the graph.	
	
	Let $G' = (V',E')$ be the graph passed to a recursive call. 
	The algorithm spends $O(|E'|)$ time to test whether there are strong articulation
	points in the graph (line~\ref{alg-line:SAPs}), and additionally $O(|E'|)$ time 
	to compute the strong articulation points, to execute the $split$ operation on 
	an arbitrary strong articulation point in each strongly connected component,
	and to recompute strongly connected components (lines \ref{alg-line:2vccs}--\ref{alg-line:sccs-v}).
	Since the recursive calls are executed on subgraphs whose sets of edges are disjoint (since the $split$ operator simply partitions the edges incident to the vertex on which the operation is executed, and moreover, all the strongly connected components are disjoint), it follows that the total time spend for the above procedures in all instances at each recursion depth is $O(m_0)$.
	Notice that the number of vertices does not exceed $2m_0$, by Lemma~\ref{lemma:edges-vertices-split}, after any sequence of split operations, and thus this time bound holds for every recursion depth.

	Let $G'$ be the graph at some recursive call.
	We show that the graph passed to each subsequent recursive call has at 
	most $\max\{|E(G')|-\sqrt{m_0}, 2\sqrt{m_0}\}$ edges, or $G'$ is a 
	$2$-vertex-connected subgraph and thus the recursion stops. 
	This implies a recursion depth of $O(\sqrt{m_0})$ as follows.
	If a graph passed to a recursive call has at most $2\sqrt{m_0}$ edges, it means that also the number of vertices is at most $2\sqrt{m_0}$.
	Therefore, even if the algorithm simply identifies a strong articulation point,
	executes the $split$ operation, and recurses on each strongly connected 
	component of the resulting graph, the total recursion depth is at 
	most $O(\sqrt{m_0})$.
	On the other hand, there can be at most $\sqrt{m_0}$ cases where the graph 
	that is passed in a recursive call has $\sqrt{m_0}$ fewer edges that $G'$.
	Overall, this proves that the recursion depth is bounded by $O(\sqrt{m_0})$.
	
	It remains to show the claimed bound on the size of the graph passed to 
	a recursive call in line~\ref{alg-line:recursion-big-sets-v}.
	By Lemma~\ref{lemma:sccs-after-split} every \VertexIsolatedOut{1} (resp., \VertexIsolatedIn{1}) component will be in a separate strongly connected component with at most $2\sqrt{m_0}$ edges.
	Now, let $C$ be the set of vertices that were not included in any \VertexIsolatedOut{1} or any \VertexIsolatedIn{1} component.
	This set did not contain any \VertexIsolatedOut{1} or any \VertexIsolatedIn{1} component~$S$ with less than $\sqrt{m_0}$ edges in $G'$ since otherwise such a set~$S$ would contain a vertex~$x$ that lost an edge (and thus was added to~$L$) and the algorithm would search for a \VertexIsolatedOut{1} or a \VertexIsolatedIn{1} component of~$x$, identifying~$S$ in this way.
	This means, by Lemma~\ref{lemma:large-set-will-break-vertex}, that $C$ either is a $2$-vertex-connected subgraph, or there are two disjoint sets in $A,B\subset C$, $|E'(A)|, |E'(B)| > \sqrt{m_0}$ that are either not strongly connected to each other 
	or separated by at most one strong articulation point in $G'[C]$.
	If the later holds, $A$ and $B$ will be separated in line~\ref{alg-line:2vccs}, and every graph passed to a subsequent recursive call has at most $\max\{|E'(G')|-\sqrt{m_0}, 2\sqrt{m_0}\}$ edges.\proofend
\end{proof}

The following theorem summarizes the result of this section.

\begin{theorem}\label{th:2vcs}
	The maximal $2$-vertex-connected subgraphs of a digraph can be computed in $O(m^{3/2})$~time.
\end{theorem}
 
\makeatletter{}\section{$k$-edge-connected subgraphs}\label{sec:kedge}

In this section we extend our algorithm for $2$-edge-connected subgraphs 
to $k$-edge-connected subgraphs for $k > 2$. At the end of the section we discuss how 
to obtain better running time bounds for undirected graphs.

Let $G = (V,E)$ be a digraph and $u \in V$ be a vertex. Let for this section 
$k' = k-1$.
We define a \EdgeIsolatedOut{k'} component $S$ of $u$ to be a subgraph of $G$ with $\widetilde{k} \le k'$ outgoing edges such that $S$ contains $u$ and there is no component
$S'\subset S$ containing $\widetilde{k}$ or fewer outgoing edges. 
We denote a \EdgeIsolatedOut{k'} component of $u$ by $\fEdgeIsoOut{k'}(u)$.
Analogously, we define a \EdgeIsolatedIn{k'} component $S$ of $u$ to be a 
subgraph of $G$ with $\widetilde{k} \le k'$ incoming edges, such that $S$ contains $u$ 
and there is no component $S'\subset S$ containing $\widetilde{k}$ or fewer incoming edges. 

We would like to extend the algorithm for maximal $2$-edge-connected subgraphs
to maximal $k$-edge-connected subgraphs for $k \ge 2$. 
For this we need to identify 
\EdgeIsolatedOut{k'} components (and therefore \EdgeIsolatedIn{k'} components
by using the reverse graph) for a given vertex~$u$
in time proportional to their size,
potentially with an additional factor depending on $k$. 
The first idea would be to simply start $k = k'+1$ depth-first searches from~$u$.
Assume for now that the first $k'$ searches $F_1, \ldots, F_{k'}$ each visited
$\ell(k', \Delta)$ edges, for some function~$\ell(k', \Delta)$ of order 
$O(k' \Delta)$ specified later, and 
let $T_1, \ldots, T_{k'}$ denote the DFS trees generated by
the searches, respectively. 
Further assume that there exists a \EdgeIsolatedOut{k'}
component~$\fEdgeIsoOut{k'}(u)$ that contains at most $\Delta$ edges. 
Suppose we have for each $1 \le i \le k'$ a path $P_i$ in $T_i$ that starts at $u$
and ends outside of $\fEdgeIsoOut{k'}(u)$.
The idea of the $k = k'+1$ searches from~$u$ is that the 
first $k'$ searches should each reduce the number of outgoing edges of
$\fEdgeIsoOut{k'}(u)$ by one by reversing the direction of the edges on $P_i$
after the $i$-th such search, and the $(k'+1)$-st search should explore
exactly the edges of $\fEdgeIsoOut{k'}(u)$ (and does not have to be a DFS). 
However, since there are multiple edges leaving $\fEdgeIsoOut{k'}(u)$, each DFS~$F_i$
might enter and leave $\fEdgeIsoOut{k'}(u)$ multiple times and thus we cannot determine
such paths $P_i$ that end outside of $\fEdgeIsoOut{k'}(u)$ so easily.
Note that if we reverse a path that ends inside of 
$\fEdgeIsoOut{k'}(u)$, then the number of outgoing edges of $\fEdgeIsoOut{k'}(u)$
remains the same and no progress was made by this search.
We first show that under the assumption that we have for all $1 \le i \le k'$
a path $P_i$ from~$u$ that ends outside of $\fEdgeIsoOut{k'}(u)$, 
the strategy of reversing the path $P_i$ before conducting
the search $F_{i+1}$ works, i.e., we extend Section~\ref{sec:edgeIsolatedOne} to 
this case. We then construct $O(k')$ paths for each of the searches of 
which one of them is guaranteed to end outside of the \EdgeIsolatedOut{k'}
component of $u$ (if such a component with at most $\Delta$ edges exists).
Based on this, we provide an algorithm for computing the maximal $k$-edge-connected
subgraphs with exponential dependence on $k = k' + 1$.

\subsection{(k-1)-edge-out components}\label{sec:edgeIsolatedK}

Recall $k' = k-1$. 
Given an integer $\Delta$, we assume that the starting vertex~$u$ can 
reach at least $\ell(k', \Delta)$ edges for some $\ell(k', \Delta) \in O(k' \Delta)$
with $\ell(k', \Delta) < m$. 
Notice that if $u$ cannot visit $\ell(k', \Delta)$ 
edges, then the reachable subgraph from $u$ defines a \EdgeIsolatedOut{k'} 
component of $u$ containing less than $\ell(k', \Delta)$ edges. 
Let $F_1$ denote a depth-first search for up to $\ell(k', \Delta)$ edges
started from $u$ in $G_1 = G$. Let $P_1$ be a path from $u$ in $F_1$.
Let $G_{i+1}$ for $1 \le i \le k'$ be defined as $G_i$ with the edges of $P_i$
reversed, where $F_i$ is the depth-first search conducted on $G_i$ 
(from $u$, for up to $\ell(k', \Delta)$ edges) and $P_i$ is a path 
from $u$ in $F_i$. We can interpret the graph $G_{i+1}$ as residual graph after
sending one unit of flow along the path $P_i$ in the graph $G_i$. The following
lemma supports this interpretation by showing that if $\widetilde{k}$~paths end in
a vertex set $T$ with $u \not\in T$, then the number of edges from $S = V \setminus T$
to $T$ in $G_{i+1}$ is reduced by $\widetilde{k}$ compared to $G$.
\begin{lemma}\label{lemma:fewer-out}
	Let $S$ be a set of vertices containing $u$. Let $T = V \setminus S$
	be a set of vertices that contains $0 \le \widetilde{k} \le i$ of the endpoints 
	of the paths $P_1, \ldots, P_i$ for some $i \le k'$. 
	Then there are $\widetilde{k}$ fewer edges from $S$ to $T$ in $G_{i+1}$ than in $G$.
\end{lemma}
\begin{proof}
	Consider the (multi-)graph $G'$ that is constructed from $G$ by contracting 
the vertices of $S$ to a single vertex $s$ and the vertices of 
$T$ to a single vertex~$t$. Applying the contraction to the paths $P_j$,
we obtain for each $1 \le j \le i$ a set of edges $E'_j$ between $s$ and $t$
that represent the contracted path, where we keep the direction of the edges 
as in $P_j$. Let $G'_1 = G'$ and let $G'_{j+1}$ for $1 \le j \le i$ be the 
(multi-)graph obtained from $G'_j$ by reversing the edges of $E'_j$, i.e., 
the graph $G'_j$ can be obtained from $G_j$ by contracting $S$ and $T$, respectively.
By definition, the graphs $G'_j$ differ from $G'$ only 
in the direction of the edges between $s$ and~$t$.
 Further we have that if $P_j$ ends 
at a vertex of $T$ (case~1), then the number of edges from $s$ to $t$ in $E'_j$ 
is one more than the number of edges from $t$ to $s$; in contrast, if $P_j$ ends 
at a vertex of $S$ (case~2), there are as many edges from $s$ to $t$ as 
from $t$ to $s$ in $E'_j$. In case~1 the number of edges from $s$ to $t$ in
$G'_{j+1}$ is one lower than in $G'_j$, while in case~2 the number of edges 
from $s$ to $t$ is the same in $G'_j$ and $G'_{j+1}$.
Let $0 \le \widetilde{k} \le i$ be the number of paths
of $\{P_1, \ldots, P_i\}$ that end in $T$. We have that
the number of edges from $s$ to $t$ in 
$G'_{i+1}$, and therefore from $S$ to $T$ in $G_{i+1}$, is equal to 
the number of paths from $S$ to $T$ in $G$ minus $\widetilde{k}$.
\proofend\end{proof}

The following lemma shows that for the first of the $k'+1$ searches that cannot 
reach $\ell(k', \Delta)$ edges from~$u$, we have that the subgraph traversed 
by this search induces a \EdgeIsolatedOut{k'} component of $u$.

\begin{lemma}\label{lemma:keout-soundness}
	Let the $i$-th search of the searches $F_1, \ldots, F_{k'+1}$ be the first 
	one that visits less than $\ell(k', \Delta)$ edges if such a search 
	with $i > 1$ exists. Then there exists an \EdgeIsolatedOut{(i-1)} component
	of $u$ with less than $\ell(k', \Delta)$ edges and the subgraph 
	induced by the vertices traversed by $F_i$ defines this component.
\end{lemma}

\begin{proof}
Let $S$ be the set of vertices traversed by $F_i$ and let $T$ be $V \setminus S$.
Note that $T$ contains at least one vertex that was traversed by some search $F_1, 
\ldots, F_{i-1}$ but not by $F_i$ because $F_i$ could traverse $\ell(k', \Delta)$
edges if it could traverse the same vertices as the other searches.
Notice further that $S$ contains $u$. By the definition of $S$ and the assumption 
that $F_i$ traverses less than $\ell(k', \Delta)$ edges, we have that there are no
edges from~$S$ to~$T$ in~$G_i$. Thus by Lemma~\ref{lemma:fewer-out} the number of 
edges from $S$ to $T$ in $G$ is equal to the number~$\widetilde{k}$ of paths of 
$P_1, \ldots, P_{i-1}$ that end in $T$. Thus $S$ has $\widetilde{k} < i$ outgoing edges in $G$. 

To show that $S$ is an \EdgeIsolatedOut{(i-1)} 
component of~$u$, it remains to show that $S$ does not contain a proper 
subset $\hat{S}$
that contains $u$ and has $\widetilde{k}$ or less outgoing edges. Assume by contradiction
such a set $\hat{S}$ exists and let $\hat{T} = V \setminus \hat{S}$. By $\hat{T} 
\supseteq T$ at least $\widetilde{k}$ paths of the paths $P_1, \ldots, P_{i-1}$
end in $\hat{T}$. Thus by Lemma~\ref{lemma:fewer-out} 
there are at least $\widetilde{k}$ edges from $\hat{S}$
to $\hat{T}$ in $G$. If there were exactly $\widetilde{k}$ edges from $\hat{S}$ to $\hat{T}$ 
in $G$, then there would be no edges from $\hat{S}$ to $\hat{T}$ in $G_i$. 
Recall that the search $F_i$ is conducted in the graph $G_i$.
Thus this is a contradiction to $S$ being the set of vertices explored by $F_i$ 
from $u$.
\proofend\end{proof}

The following lemma shows that identifying a \EdgeIsolatedOut{k'} component of $u$ 
reduces to identifying paths $P_1, \ldots, P_{k'}$ that all end outside of the component.

\begin{lemma}\label{lemma:keout-completeness}
	Let $\fEdgeIsoOut{k'}(u)$ be a \EdgeIsolatedOut{k'} component of $u$ with 
	$\lvert \fEdgeIsoOut{k'}(u) \rvert \le \Delta$.
	Assume all paths $P_1, \ldots, P_{k'}$ end outside of $\fEdgeIsoOut{k'}(u)$ and 
	all searches $F_1, \ldots, F_{k'}$ have visited $\ell(k', \Delta)$ edges.
	Then the subgraph traversed by $F_{k'+1}$ is $\fEdgeIsoOut{k'}(u)$.
\end{lemma}

\begin{proof}
	Let $S$ denote the set of
	vertices of $\fEdgeIsoOut{k'}(u)$ and let $T = V \setminus S$.
	 By definition $S$ has at most $k'$ outgoing edges in $G$ and no proper 
	 subset of $S$ has $k'$ or less outgoing edges. By assumption all paths 
	 $P_1, \ldots, P_{k'}$ end in $T$. Thus by Lemma~\ref{lemma:fewer-out}
	 there are no edges from $S$ to $T$ in $G_{k'+1}$ and hence $F_{k'+1}$ 
	 traverses a subset of~$S$. 
	 
	 It remains to show that $F_{k'+1}$ traverses exactly the vertices of $S$.
	 Assume by contradiction it traverses a proper subset $S'$ of $S$.
	 By Lemma~\ref{lemma:keout-soundness} this implies that the subgraph
	 induced by $S'$ is a \EdgeIsolatedOut{k'} component of $u$, a contradiction
	 to the minimality of~$S$.
\proofend\end{proof}

\subsubsection{Finding outgoing DFS-tree paths.}
Assume a \EdgeIsolatedOut{k'} component~$\fEdgeIsoOut{k'}(u)$ of $u$ with 
$\lvert \fEdgeIsoOut{k'}(u) \rvert \le \Delta$ exists. We show next how to find a set 
of $O(k)$ paths $\mathcal{P}_1$ from $u$ in the DFS tree~$T_1$ such that at least one
of the paths ends outside of~$\fEdgeIsoOut{k'}(u)$. We also specify the 
number of edges~$\ell(k', \Delta)$ for which we conduct each depth-first search
in order to ensure this property.
As we do not know which of the paths of $\mathcal{P}_1$ ends outside, 
we define $O(k)$ variants of the graph $G_2$, one for each of the paths 
in $\mathcal{P}_1$, i.e., the $j$-th variant of $G_2$ is equal to $G_1 = G$ with the $j$-th path of 
$\mathcal{P}_1$ reversed. Provided that the DFS $F_1$ has visited $\ell(k', \Delta)$
edges, we start a second DFS search~$F_2$ on each of the variants of $G_2$. 
In the same manner, we start $O(k)$ third DFS searches for each path in $\mathcal{P}_i$ for 
each variant of $G_2$ and so on, that is, we have $O(k^2)$ variants of $G_3$
and $O(k^k)$ variants of $G_{k'+1}$. Thus instead of $k = k'+1$ DFS searches, 
we perform $O(k^k)$ DFS searches in total in order to ensure that at least in one 
variant all the paths $P_1, \ldots, P_{k'}$ end outside of $\fEdgeIsoOut{k'}(u)$ and 
thus the $(k'+1)$-st search in this variant of $G_{k'+1}$ explores 
$\fEdgeIsoOut{k'}(u)$ by Lemma~\ref{lemma:keout-completeness}.

To identify a path from $u$ that ends outside of $\fEdgeIsoOut{k'}(u)$,
we conduct each DFS from $u$ in chunks of $\Delta + 1$ edges. After each chunk
we identify one candidate path from $u$ in the DFS tree and make a 
recursive call that starts the next DFS on the graph with the candidate path
reversed. We show that either the candidate path indeed ends outside of 
$\fEdgeIsoOut{k'}(u)$ or the DFS has either traversed one additional edge 
of the outgoing edges of $\fEdgeIsoOut{k'}(u)$ or retracted along one 
of these outgoing edges. Recall that whenever a DFS traversal retracts
from some vertex~$v$ to a proper ancestor of~$v$, then it will not visit~$v$
or any of its outgoing edges ever again. Thus in particular we have 
for every outgoing edge of $\fEdgeIsoOut{k'}(u)$ that it is explored only 
once and that the DFS retracts along the edge at most once. Hence both cases
together can happen in at most $2 k'$ of the chunks and thus by continuing the DFS
for $2k' + 1$ chunks we are guaranteed to have identified at least one path from 
$u$ to a vertex outside of $\fEdgeIsoOut{k'}(u)$. By this argument we implicitly set 
$\ell(k', \Delta)$ to $(2k'+1) (\Delta + 1)$. The following lemma shows 
that this strategy indeed either makes progress by using an outgoing edge or 
identifies a path ending outside. For each chunk the candidate path 
is given by the path from $u$ to the highest vertex in the DFS tree (i.e., the 
vertex closest to the root~$u$) that is visited during this chunk.

\begin{lemma}\label{lemma:dfs-extend}
	Let $F$ denote a \textup{(}possibly empty\textup{)}
	DFS started at $u$.	Assume there are at least 
	$\Delta + 1$ edges reachable from $u$ that where not explored 
	by $F$. Let $F'$ be the DFS obtained by extending $F$ by $\Delta + 1$ edges.
	\emph{At least one} of the following statements is true.
	\begin{enumerate}[\textup{(}a\textup{)}]
		\item The nearest common ancestor \textup{(}NCA\textup{)} in the DFS tree
		of all vertices visited by $F' \setminus F$
		is not in $\fEdgeIsoOut{k'}(u)$.\label{dfs-extend-nca}
		\item The DFS $F'$ explored at least one more outgoing edge of 
		$\fEdgeIsoOut{k'}(u)$ than the DFS~$F$.\label{dfs-extend-moreout}
		\item The DFS retracted along an outgoing edge of $\fEdgeIsoOut{k'}(u)$
		while extending $F$ to $F'$.\label{dfs-extend-retract}
	\end{enumerate}
\end{lemma}
\begin{proof}
	Let $s$ be the vertex where the DFS~$F$ stopped and the extension of $F$
	starts. Let $h$ be the nearest common ancestor in the DFS tree of all vertices
	visited during extending $F$ to $F'$. Note that $h$ was already explored by $F$.
	If the DFS retracts to an ancestor of $s$ at some point during the extension, 
	then $h$ is equivalent to the highest vertex the DFS retracts to; 
	otherwise we have $h = s$.

	First consider the case where $s$ is not in $\fEdgeIsoOut{k'}(u)$.
	Then either also $h$ is not in $\fEdgeIsoOut{k'}(u)$ and thus 
	(\ref{dfs-extend-nca}) is satisfied or the DFS retracts back to some ancestor
	of $s$ that is inside of $\fEdgeIsoOut{k'}(u)$, which satisfies (\ref{dfs-extend-retract}).

	Assume now $s \in \fEdgeIsoOut{k'}(u)$.
	Since $|\fEdgeIsoOut{k'}(u)| \leq \Delta$ and the DFS $F'$ visits $\Delta+1$
	additional edges, it follows that $F'$ visits at least one additional
	edge that is not in $\fEdgeIsoOut{k'}(u)$. To reach edges outside of 
	$\fEdgeIsoOut{k'}(u)$ from $s$, 	the DFS~$F'$ can either (1) use one of the outgoing
	edges of $\fEdgeIsoOut{k'}(u)$ that were not traversed by $F$ 
	or (2) complete the DFS traversal from $s$
	and retract to some ancestor of $s$ in the current DFS tree. 
	In case (1) the statement (\ref{dfs-extend-moreout}) holds.
	In case (2) we distinguish three sub-cases. Recall that $h$ is the highest vertex 
	the DFS retracts to and the nearest common ancestor of all vertices visited 
	during the extension of the DFS.
	\begin{enumerate}[(i)]
		\item If $h$ is in $\fEdgeIsoOut{k'}(u)$ but some vertex on the tree path from 
	$h$ to $s$ is not in $\fEdgeIsoOut{k'}(u)$, then (\ref{dfs-extend-retract}) is
	satisfied.\label{c1}
		\item  If all vertices on the path from $h$ to $s$ in the DFS tree
	are in $\fEdgeIsoOut{k'}(u)$, then the DFS has 
	to use at least one outgoing edge of $\fEdgeIsoOut{k'}(u)$ to visit edges outside of 
	$\fEdgeIsoOut{k'}(u)$ and thus (\ref{dfs-extend-moreout}) holds. \label{c2}
		\item If neither (\ref{c1}) nor (\ref{c2}) holds, then $h$ is not in
		$\fEdgeIsoOut{k'}(u)$ and hence (\ref{dfs-extend-nca}) holds. \proofend
	\end{enumerate} 
\end{proof}

In Procedure~\ref{proc:kEOut} we combine the results of this section to an algorithm
that returns a \EdgeIsolatedOut{k'} component of $u$ whenever one with at most $\Delta$ 
edges exists, and might return the empty set if no such component exists.  The 
soundness of the algorithm is given by Lemma~\ref{lemma:keout-soundness}, i.e.,
whenever one of the subsequent $k'+1$ 
DFS searches can visit less than $\ell(k', \Delta) = 
(2k' + 1) (\Delta + 1)$ edges, then the set of vertices visited by this search 
induces a \EdgeIsolatedOut{k'} component of $u$. The completeness
of the algorithm is as follows: by Lemma~\ref{lemma:keout-completeness}
the $(k'+1)$-st search identifies \EdgeIsolatedOut{k'} component of $u$
with at most $\Delta$ edges given that the $k'$ paths $P_1, \ldots, P_{k'}$
all end outside of the component (where a path $P_i$ is identified in the graph 
constructed from $G_{i-1}$ by reversing the edges of $P_{i-1}$); 
and by Lemma~\ref{lemma:dfs-extend}, and the observation that the cases 
(\ref{dfs-extend-moreout}) and (\ref{dfs-extend-retract}) of 
Lemma~\ref{lemma:dfs-extend} can each happen at most 
$k'$ times if a \EdgeIsolatedOut{k'} component of $u$ with at most $\Delta$ edges
exists, this property is satisfied by
the paths constructed in at least one of the  sequences of depth-first searches initialized by the 
recursive calls of the algorithm. 
One DFS search takes time $(2k' + 1) (\Delta + 1)$ and makes at most 
$2k' + 1$ recursive calls. In Procedure~\ref{proc:kEOut} the recursion 
depth is explicitly bounded by $k = k' + 1$ (line~\ref{alg-line:boundrecdepth}).
Thus we have shown the following lemma.
\begin{lemma}\label{lemma:kEOut}
	We can compute in $O((2k)^{k+1} \cdot \Delta)$ time a \EdgeIsolatedOut{(k-1)}
	component of a vertex~$u$ that contains less than
	$(2k-1)(\Delta+1)$ edges, or otherwise we
	conclude that there is no \EdgeIsolatedOut{(k-1)} component of~$u$ containing at 
	most $\Delta$~edges.
\end{lemma}

\begin{procedure}[t!]
	\caption{kEOut($G$, $u$, $\Delta$)}
	\label{proc:kEOut}
	\DontPrintSemicolon
	\KwIn{Digraph $G=(V,E)$, a vertex $u$, and an integer $\Delta$}
	\KwOut{For $k' = k-1$ either a vertex set inducing a
	\EdgeIsolatedOut{k'} component of $u$ with less than
	$(2k' + 1)(\Delta+1)$ edges or the empty set; if the empty set is returned, no 
	\EdgeIsolatedOut{k'} component of $u$ with less than $\Delta$ edges exists}
	\BlankLine
	
	Initialize DFS $F$ starting from~$u$\;
	\For{$2k' + 1$ times}
	{
		Extend the DFS $F$ for at most $\Delta + 1$ edges\;
		\If{at most $\Delta$ edges added to $F$}{
			\Return{the set of vertices explored by $F$}
		}\ElseIf{the recursion depth is at most $k'$\label{alg-line:boundrecdepth}}{
			Let $h$ be the NCA in the DFS tree of the vertices visited during the extension\;
			Let $G'$ be $G$ with the DFS tree path from $u$ to $h$ reversed\;
			$S \gets $\ref{proc:kEOut}($G'$, $u$)\;
			\If{$S \ne \emptyset$}{
				\Return $S$
			}
		}
	}
	\Return{$\emptyset$}
\end{procedure}

\subsection{Computing the $k$-edge-connected subgraphs of a digraph.}
	\label{sec:k-edge-connected-strong-components}	
	
	Our algorithm to compute the maximal $k$-edge connected subgraphs of a given 
	digraph follows the same structure as the algorithm given in 
	Section~\ref{sec:2-edge-connected-strong-components} for $k=2$. The main 
	difference lies in the different subroutine we use to determine \EdgeIsolatedOut{(k-1)}
	components with at most $\Delta$ edges. 
	For $k > 2$ we use Procedure~\ref{proc:kEOut}, which has a running time exponential
	in $k$ and linear in~$\Delta$. The increased time to determine an edge
	cut of at most $k-1$ edges (i.e., lines~\ref{kl:cut} and \ref{kl:cuts-in-sccs}
	take time $O(m \log n)$~\cite{gabow1995matroid})
	leads to an additional factor of $\log n$ in the running time.
	Thus we obtain for any constant~$k > 2$ an $O(m^{3/2} \log n)$ time algorithm.
	
	The basic algorithm for maximal $k$-edge connected subgraphs finds
	for each strongly connected component of the input graph a (directed) cut 
	of at most $k-1$ edges if one exists, 
	removes the cut edges from the graph, and recurses
	on each strongly connected component of the remaining graph. A cut of at most 
	$k - 1$ edges can be found with Gabow's algorithm~\cite{gabow1995matroid}
	in $O(k m \log{n})$ time and whenever a cut is found at least two vertices
	seize to be strongly connected, thus this algorithm takes time 
	$O(k m n \log{n})$ time.
	
	To improve upon the basic algorithm for constant $k$ and
	sparse graphs with $\sqrt{m} < n$, we search for \EdgeIsolatedOut{(k-1)}
	components with at most $\Delta = \sqrt{m}$ edges from all vertices that 
	have lost adjacent edges in a prior iteration of the algorithm, using 
	Procedure~\ref{proc:kEOut}. Analogously, 
	we search for \EdgeIsolatedIn{(k-1)} components with at most $\sqrt{m}$
	edges from these vertices by applying  Procedure~\ref{proc:kEOut} on the 
	reverse graph. Note that whenever a cut with at most $k-1$ edges
	exists, then one side of the cut contains a \EdgeIsolatedOut{(k-1)} component 
	and the other side of the cut contains a \EdgeIsolatedIn{(k-1)} component, and 
	vice versa. Further, 
	if a subgraph is a \EdgeIsolatedOut{(k-1)} component in a recursive call 
	of the algorithm but was not before the recursive call, then this subgraph
	must have had at least one additional outgoing edge before the recursive 
	call. Thus by searching from vertices that have lost adjacent edges,
	we ensure to find all \EdgeIsolatedOut{(k-1)} components with at most 
	$\sqrt{m}$ edges in each recursive call. Hence if no \EdgeIsolatedOut{(k-1)}
	or \EdgeIsolatedIn{(k-1)} component is found by these searches, we know
	that every cut of at most $k-1$ edges divides the graph 
	into two subgraphs with more than $\sqrt{m}$ edges each. 
	In this case we execute one iteration of the basic algorithm.
	Since subgraphs of more than $\sqrt{m}$ edges can be removed from $G$ at 
	most $\sqrt{m}$ times, we can bound the recursion depth of the algorithm 
	by $\sqrt{m}$ as in Section~\ref{sec:2edge}.
The following lemmata formalize this discussion and are straightforward generalizations of the $k=2$ case.
\begin{lemma}
Let $C$ be a set of vertices in $G$.
		Every \EdgeIsolatedOut{(k-1)} or \EdgeIsolatedIn{(k-1)} component 
		\textup{(}of some vertex $u\in C$\textup{)} in $G[C]$ that is not such a component in~$G$ must contain an endpoint of an edge incident to $G[C]$.
		Moreover, if there is no \EdgeIsolatedOut{(k-1)} or 
 		\EdgeIsolatedIn{(k-1)} component containing at most $\Delta$ edges
 		for any vertex $u \in C$ in $G[C]$, 
 		then one of the following holds:
		\begin{enumerate}[\textup{(}a\textup{)}]
			\item $G[C]$ is a $k$-edge-connected subgraph of $G$.\label{sublem:kec}
			\item There are two sets $A,B\subset C$ with $|E(G[A])|, |E(G[B])| > \Delta$
			such that $A$ and $B$ are in different strongly connected components of $G[C]$.
			\label{sublem:klargesccs}
		\item For each cut of size at most $k-1$ in $G[C]$ there are two sets 
		$A,B\subset C$ with $|E(G[A])|, |E(G[B])| > \Delta$ that get disconnected by 
		the deletion of the cut edges.\label{sublem:largecut}
		\end{enumerate}	\label{lemma:large-set-will-break-k}
	\end{lemma}
	\begin{proof}
		Let $k' = k-1$.
		We first show that every \EdgeIsolatedOut{k'} component $\fEdgeIsoOut{k'}(u)$ of some vertex $u \in C$ that is not a \EdgeIsolatedOut{k'} component
		in $G$ must contain an edge $(x,y)$ with $x\in \fEdgeIsoOut{k'}(u)$ and $y \notin C$.
		Assume, by contradiction, that $\fEdgeIsoOut{k'}(u)$ exists but there is no 
		such edge $(x,y)$ in $G$ with $x\in \fEdgeIsoOut{k'}(u)$ and $y \notin C$.
		In this case we have that the very same component $\fEdgeIsoOut{k'}(u)$ 
		is a \EdgeIsolatedOut{k'} component of $u$ in $G$.
		The same argument on the reverse graph shows that every new \EdgeIsolatedIn{k'} component (of some vertex $u\in C$) in $G[C]$ must contain an endpoint of an edge incident to $G[C]$.
		
		Now we turn to the second part of the lemma.
		If $G[C]$ is strongly connected and does not contain a cut of size at most $k'$,
 		then $G[C]$ is $k$-edge-connected and thus (\ref{sublem:kec}) holds.
 		If $G[C]$ is not strongly connected, then it contains (at least)
 		two disjoint sets $A, B \subset C$ such that both $G[A]$ and $G[B]$
 		are strongly connected components of $G[C]$ and 
 		$G[A]$ has no outgoing edge in $G[C]$ (i.e., $G[A]$ is a sink in the DAG 
 		of SCCs of $G[C]$) and $G[B]$ has no incoming edge in $G[C]$ (i.e., $G[B]$ is a 
 		source in the DAG of SCCs of $G[C]$).
 		That is, in $G[C]$ we have that $G[A]$ contains a \EdgeIsolatedOut{k'} component of some $u \in C$
 		and $G[B]$ contains a \EdgeIsolatedIn{k'} component of some $u' \in C$. 
		Both $G[A]$ and $G[B]$ may have the same property in $G$ or
		be new such components in $G[C]$ compared to $G$.
 		In any case it contradicts the assumptions if one of them has at most $\Delta$ 
 		edges and otherwise statement (\ref{sublem:klargesccs}) holds.
 		If $G[C]$ is strongly connected and contains a cut of with at most $k'$
 		edges, an 
 		analogous argument can be made for two disjoint sets $A, B \subset C$
 		by considering the DAG of SCCs of $G[C]$ with the cut edges 
 		removed. If the number of cut edges is minimal, we have that
 		the cut edges are the only incoming edges of $B$ and the 
 		only outgoing edges of $A$ in $G[C]$.
 		We have that case~(\ref{sublem:largecut})
 		holds if the assumptions of the lemma are satisfied.
	\proofend\end{proof}
	
	\begin{algorithm}[t!]
		\DontPrintSemicolon
		\KwIn{Strongly connected digraph $G=(V,E)$ and a list of vertices $L$ (initially $L = V$)}
		\KwOut{The $k$-edge-connected subgraphs of $G$}
		\BlankLine
		Let $m_0$ be the number of edges in the initial graph\;
		\If{$G$ has no cut of less than $k$ edges\label{kl:cut}}{
			\Return $\{G\}$ as $k$-edge connected subgraph\label{kl:output}
		}
		
		\While{$L \not = \emptyset$ and $G$ has more than $2 k \sqrt{m_0}$ edges}
						{
			Extract a vertex $u$ from $L$
			
			$S \gets $\ref{proc:kEOut}($G$, $u$, $\sqrt{m_0}$)\label{kl:search1}\;
			$S^R \gets $\ref{proc:kEOut}($G^R$, $u$, $\sqrt{m_0}$)\label{kl:search2}\;
			
			If either $S$ or $S^R$ is not empty, remove from $G$ all the incident edges to one non-empty set of $S$ and $S^R$ and add their endpoints to $L$ \label{kl:found}
		}
		Compute SCCs $C_1, \ldots, C_c$ of $G$\label{kl:sccs1}\;
		$U \gets \emptyset$\;
		\ForEach{$C_i, 1 \leq i \leq c$}{
			Remove a $(k-1)$-cut from $G[C_i]$ (if it exists) \label{kl:cuts-in-sccs}\;
			Recompute SCCs and delete the edges between them\label{kl:sccsdel}\;
			\ForEach{SCC $C'$}{
				Insert into $L'$ the vertices of $C'$ that are endpoints of newly 
				deleted edges\;
				$U \gets U \cup kECS(C',L')$\label{kl:recursion-big-sets}\;
			}
		}
		\Return $U$
		\caption{\textsf{$kECS(G, L)$}}
		\label{alg:kECSC}
	\end{algorithm}
	
	\begin{lemma}
		The algorithm $kECS$ is correct.
	\end{lemma}
 	\begin{proof}
		Whenever the algorithm $kECS$ reports a $k$-edge-connected subgraph in 
		line~\ref{kl:output}, then it is a strongly connected subgraph 
		that does not contain any cut with at most $k-1$ edges, which is by 
		definition a $k$-edge-connected subgraph.
		Thus it suffices to show that $kECS$ reports all the maximal $k$-edge-connected 
		subgraphs. Notice that this also implies that the reported $k$-edge-connected 
		subgraphs are maximal.
		Let $C$ be a maximal $k$-edge-connected subgraph.
		We show that the vertices of $C$ do not get separated by the algorithm, and
		therefore $C$ is reported eventually as a $k$-edge-connected subgraph.
		Since there are $k$ edge-disjoint paths between every pair of vertices in~$C$, 
		any search for either a \EdgeIsolatedOut{(k-1)} or a \EdgeIsolatedIn{(k-1)}
		component of a vertex $u$ (lines~\ref{kl:search1}--\ref{kl:search2}) 
		either returns a superset of~$C$ or 
		fails to identify such a set containing a subset of the vertices of~$C$.
		Furthermore, notice that any deletion of an edge that does not have both endpoints 
		in $C$ does not affect the fact that $C$ is $k$-edge-connected.
		That is, unless an edge with both endpoints in~$C$ is deleted, 
		no cut with at most $k-1$ edges appears in $C$.
		Thus, it remains to show that no edge $(x,y)$ such that $x, y \in C$ is 
		ever deleted throughout the algorithm.
		The edges deleted in line~\ref{kl:found} of the algorithm are 
		incident to a \EdgeIsolatedOut{(k-1)} or a \EdgeIsolatedIn{(k-1)} component.
		Since $C$ is always fully inside or fully outside of such a set, 
		no edge from $C$ is deleted.
		The edges deleted in line~\ref{kl:cuts-in-sccs} are cuts with at most $k-1$
		edges and the edges deleted in line~\ref{kl:sccsdel} before the recursive calls 
		are between separate strongly connected components.
		Since $C$ is $k$-edge-connected, no edges from $C$ are deleted.
		Finally, notice that at each level of recursion at least cut with at most $k-1$
		edges of each strongly connected component of the graph is deleted 
		and the algorithm is recursively executed on each resulting strongly connected
		component. Thus, the recursive calls finally consider
		all strongly connected subgraphs that do not contain 
		cuts with at most $k-1$ edges, including $C$.
	\proofend\end{proof}
	
	\begin{lemma}\label{lemma:timekecs}
		The algorithm $kECS$ runs in $O(m \sqrt{m} \log n)$ time for constant~$k > 2$.
	\end{lemma}
	\begin{proof}
		Let $\ell(k, \Delta) = (2k-1)(\Delta+1)$ and $\Delta = \sqrt{m}$,
		where $m$ is the number of edges in the input graph.
		First notice that each time we search for a \EdgeIsolatedOut{(k-1)} or a 
		\EdgeIsolatedIn{(k-1)} component, we are searching for a component with 
		$k-1$ outgoing (resp., incoming) edges containing at most $\sqrt{m}$ edges 
		or with less than $k-1$ (resp., incoming) edges and less than 
		$\ell(k, \Delta)$ edges. 
		We can identify if such a component containing a given vertex~$u$ exists
		in time $O(\sqrt{m})$ for constant~$k$ by using the 
		algorithm of Section~\ref{sec:edgeIsolatedK}.
		We initiate such a search from each vertex that appears in the list $L$ of 
		some recursive call of the algorithm. 
 		Initially, we place all vertices into the list $L$.
		Throughout the algorithm we insert into $L$ only vertices that are endpoints of deleted edges.
		Therefore, the number of vertices that are added to the lists~$L$
		throughout the algorithm is $O(m)$.
		Hence, the total time spent on these searches is $O(m\sqrt{m})$ for constant~$k$.
		
		Consider now the time spend in each recursive call without the searches 
		for  \EdgeIsolatedOut{(k-1)} and \EdgeIsolatedIn{(k-1)} components.
		Let $G'$ be the graph for which the recursive call is made and let 
		$m_{G'} = |E(G')|$.
		In each recursive call the algorithm spends $O(k m_{G'} \log n)$ time to search 
		cuts with at most $k-1$ edges in $G'$ in lines~\ref{kl:cut} 
		and~\ref{kl:cuts-in-sccs} and $O(m_{G'})$
		to compute SCCs in lines~\ref{kl:sccs1} and~\ref{kl:sccsdel}.
		Since the subgraphs of different recursive calls at the same recursion depth are
		disjoint, the total time spent at each level of the recursion is $O(m)$.
		We now bound the recursion depth for constant~$k$ with $O(\sqrt{m})$.
		
		We show that the graph passed to each recursive call has at most 
		$\max\{m_{G'}-\sqrt{m}, 2 k \sqrt{m}\}$ edges, or it is a $k$-edge-connected 
		subgraph and thus the recursion stops. This implies a recursion depth of $O(k 
		\sqrt{m})$, i.e., $O(\sqrt{m})$ for constant $k$, as follows.
		If the graph passed to a recursive call has at most $2 k \sqrt{m}$ edges,
		then also the number of vertices of this graph is at most $2 k \sqrt{m}$.
		Therefore, even if the algorithm only removes one cut from every
		strongly connected component in each recursive call, 
		the total recursion depth is at most $O(\sqrt{m})$ for constant~$k$.
		On the other hand, the number of times that the graph passed to
		a recursive call has $\sqrt{m}$ fewer edges than $G'$ is at most $\sqrt{m}$.
		Overall, this implies that the recursion depth is bounded by $O(\sqrt{m})$.
		
		It remains to show the claimed bound on the size of the graph passed to 
		a recursive call in line~\ref{kl:recursion-big-sets}.
		For every \EdgeIsolatedOut{(k-1)} or \EdgeIsolatedIn{(k-1)} component with at most 
		$2 k \sqrt{m}$ edges that is discovered throughout the algorithm, 
		its incident edges are removed and therefore it will be in a separate 
		strongly connected component with at most $2 k \sqrt{m}$ edges.
		Let $C$ be the set of vertices that were not included in any 
		\EdgeIsolatedOut{(k-1)} or \EdgeIsolatedIn{(k-1)} component. By 
		Lemma~\ref{lemma:large-set-will-break-k}
		the subgraph $G'[C]$ either is a $k$-edge-connected subgraph or 
						there are two sets $A$ and $B$ with $|E(A)|, |E(B)| > \sqrt{m}$ that will be separated in line~\ref{kl:cuts-in-sccs}.
		Thus, every graph passed to the recursive call will have at most $\max\{|E(G')|-\sqrt{m}, 2 k \sqrt{m}\}$ edges.
		The lemma follows.
	\proofend\end{proof} 
\makeatletter{}

\subsection{$k$-edge-connected
subgraphs for undirected graphs.}

The problems of computing the $k$-edge-connected 
subgraphs of an undirected graph can be reduced to the equivalent problem for directed graphs in a straightforward way.
More specifically, for a given undirected graph we construct a directed graph with the same vertex set, and replace every undirected edge with two bidirectional edges.
On the resulting digraph the set of vertices of the $k$-edge-connected
subgraphs are equivalent to the set of vertices of the $k$-edge-connected
subgraphs in the original undirected graph.

The complexity of our algorithms is determined by the choice of the parameter $\Delta$ in the algorithm that searches for \EdgeIsolatedOut{(k-1)} and the \EdgeIsolatedIn{(k-1)} components of a vertex. 
The parameter $\Delta$ determines both the depth of the recursion, which is $O(m / \Delta)$, and the time we spend searching for small components, which is $O((m+n) \Delta)$ in total.

The second factor that affects the complexity is the time spent identifying a cut
at every depth of the recursion. 
Note that the time spent searching for a cut will dominate the $O(m)$ time it takes to compute the strongly connected components before executing the recursive call.
This factor is multiplied by the maximum recursion depth in the time complexity of the algorithm.
The digraph on which we executed our algorithm originates from an undirected graph, and we can use this to search for edge cuts
of size at most $k-1$ faster.
Thus, the time complexity of our algorithms is $O(t \cdot (m /\Delta)+ n\Delta)$, where $t$ is the time required to identify a cut of size at most $k-1$ in an undirected graph.

The edge 
cuts of size at most $2$ can be identified in linear time \cite{Galil:1991,3-connectivity:ht}.
It is easy to verify that the optimal choice of $\Delta$ is therefore $m/\sqrt{n}$ for $k = 3$.
For constant $k$, we can compute an edge cut of size at most $(k-1)$ in time $O(m+n\log n)$ \cite{gabow1995matroid}.
We choose $\Delta = m / \sqrt{n}$ for $k$-edge-connected subgraphs
as well as for $3$-edge-connected 
subgraphs.
We obtain the following result.

\begin{theorem}\label{th:kecsU}
	The maximal $k$-edge-connected subgraphs of an undirected graph can be computed in $O((m + n\log n)\sqrt{n})$~time on a undirected graph with $m$ edges and $n$ vertices. For the maximal $3$-edge-connected subgraphs, our algorithm runs in $O(m \sqrt{n})$ time.
\end{theorem}

\bibliographystyle{plain}
\bibliography{references}

\begin{thebibliography}{10}

\bibitem{ChatterjeeH14}
K.~Chatterjee and M.~Henzinger.
\newblock Efficient and dynamic algorithms for alternating {B}\"uchi games and
  maximal end-component decomposition.
\newblock {\em Journal of the ACM}, 61(3):15:1--15:40, 2014.
\newblock Announced at SODA'11 and SODA'12.

\bibitem{ChechikHILP17}
Shiri Chechik, Thomas Dueholm~Hansen, Giuseppe~F. Italiano, Veronika
  Loitzenbauer, and Nikos Parotsidis.
\newblock {Faster Algorithms for Computing Maximal 2-Connected Subgraphs in
  Sparse Directed Graphs}.
\newblock In {\em {SODA}}, pages 1900--1918, 2017.

\bibitem{EppsteinGIN97}
D.~Eppstein, Z.~Galil, G.~F. Italiano, and A.~Nissenzweig.
\newblock Sparsification---a technique for speeding up dynamic graph
  algorithms.
\newblock {\em Journal of the ACM}, 44(5):669--696, September 1997.
\newblock Announced at FOCS'92.

\bibitem{erusalimskii1980bijoin}
Y.~M. Erusalimskii and G.~G. Svetlov.
\newblock Bijoin points, bibridges, and biblocks of directed graphs.
\newblock {\em Cybernetics and Systems Analysis}, 16(1):41--44, 1980.

\bibitem{even1975algorithm}
S.~Even.
\newblock An algorithm for determining whether the connectivity of a graph is
  at least k.
\newblock {\em SIAM Journal on Computing}, 4(3):393--396, 1975.

\bibitem{EvenT75}
S.~Even and Robert~E. Tarjan.
\newblock {Network Flow and Testing Graph Connectivity}.
\newblock {\em {SIAM} Journal on Computing}, 4(4):507--518, 1975.

\bibitem{gabow1995matroid}
H.~N. Gabow.
\newblock A matroid approach to finding edge connectivity and packing
  arborescences.
\newblock {\em Journal of Computer and System Sciences}, 50(2):259--273, 1995.

\bibitem{gabow2006using}
H.~N. Gabow.
\newblock Using expander graphs to find vertex connectivity.
\newblock {\em Journal of the ACM (JACM)}, 53(5):800--844, 2006.

\bibitem{GabowT85}
H.~N. Gabow and R.~E. Tarjan.
\newblock A linear-time algorithm for a special case of disjoint set union.
\newblock {\em Journal of Computer and System Sciences}, 30(2):209--221, 1985.

\bibitem{Galil:1991}
Z.~Galil and G.~F. Italiano.
\newblock Reducing edge connectivity to vertex connectivity.
\newblock {\em SIGACT News}, 22(1):57--61, March 1991.

\bibitem{Georgiadis10}
L.~Georgiadis.
\newblock Testing 2-vertex connectivity and computing pairs of vertex-disjoint
  \emph{s}-\emph{t} paths in digraphs.
\newblock In {\em Automata, Languages and Programming, 37th Int'l. Coll.,
  {ICALP} 2010, Bordeaux, France, July 6-10, 2010, Proceedings, Part {I}},
  pages 738--749, 2010.

\bibitem{2VCB}
L.~Georgiadis, G.~F. Italiano, L.~Laura, and N.~Parotsidis.
\newblock 2-vertex connectivity in directed graphs.
\newblock In {\em Proc. 42nd Int'l. Coll. on Automata, Languages, and
  Programming}, pages 605--616, 2015.

\bibitem{2ECB}
L.~Georgiadis, G.~F. Italiano, L.~Laura, and N.~Parotsidis.
\newblock 2-edge connectivity in directed graphs.
\newblock {\em ACM Trans. Algorithms}, 13(1):9:1--9:24, 2016.
\newblock Announced at SODA'15.

\bibitem{gomory1961multi}
Ralph~E Gomory and Tien~Chung Hu.
\newblock Multi-terminal network flows.
\newblock {\em Journal of the Society for Industrial and Applied Mathematics},
  9(4):551--570, 1961.

\bibitem{hariharan2007efficient}
Ramesh Hariharan, Telikepalli Kavitha, and Debmalya Panigrahi.
\newblock Efficient algorithms for computing all low st edge connectivities and
  related problems.
\newblock In {\em Proceedings of the eighteenth annual ACM-SIAM symposium on
  Discrete algorithms}, pages 127--136. Society for Industrial and Applied
  Mathematics, 2007.

\bibitem{2CC:HenzingerKL15}
M.~Henzinger, S.~Krinninger, and V.~Loitzenbauer.
\newblock Finding 2-edge and 2-vertex strongly connected components in
  quadratic time.
\newblock In {\em Proc. 42nd Int'l. Coll. on Automata, Languages, and
  Programming}, pages 713--724, 2015.
\newblock Full version available at {http://arxiv.org/abs/1412.6466}.

\bibitem{HenzingerRW17}
M.~Henzinger, S.~Rao, and D.~Wang.
\newblock Local flow partitioning for faster edge connectivity.
\newblock In {\em Proc. 28th Annual {ACM-SIAM} Symp. on Discrete Algorithms},
  pages 1919--1938, 2017.

\bibitem{Henzinger1999}
M.~R. Henzinger, V.~King, and T.~Warnow.
\newblock Constructing a tree from homeomorphic subtrees, with applications to
  computational evolutionary biology.
\newblock {\em Algorithmica}, 24(1):1--13, 1999.

\bibitem{3-connectivity:ht}
J.~E. Hopcroft and R.~E. Tarjan.
\newblock Dividing a graph into triconnected components.
\newblock {\em SIAM Journal on Computing}, 2(3):135--158, 1973.

\bibitem{Italiano2012}
G.~F. Italiano, L.~Laura, and F.~Santaroni.
\newblock Finding strong bridges and strong articulation points in linear time.
\newblock {\em Theoretical Computer Science}, 447:74--84, 2012.

\bibitem{2VCC:Jaberi2015}
R.~Jaberi.
\newblock On computing the 2-vertex-connected components of directed graphs.
\newblock {\em Discrete Applied Mathematics}, 204:164--172, 2016.

\bibitem{kanevsky1991improved}
Arkady Kanevsky and Vijaya Ramachandran.
\newblock Improved algorithms for graph four-connectivity.
\newblock {\em Journal of Computer and System Sciences}, 42(3):288--306, 1991.

\bibitem{Karger00}
D.~R. Karger.
\newblock Minimum cuts in near-linear time.
\newblock {\em Journal of the ACM}, 47(1):46--76, 2000.
\newblock Announced at STOC'96.

\bibitem{KawarabayashiT15}
K.~Kawarabayashi and M.~Thorup.
\newblock Deterministic global minimum cut of a simple graph in near-linear
  time.
\newblock In {\em Proceedings of the Forty-Seventh Annual {ACM} on Symposium on
  Theory of Computing, {STOC} 2015, Portland, OR, USA, June 14-17, 2015}, pages
  665--674, 2015.

\bibitem{Makino1988}
S.~Makino.
\newblock An algorithm for finding all the k-components of a digraph.
\newblock {\em Int'l Journal of Computer Mathematics}, 24(3-4):213--221, 1988.

\bibitem{NagamochiI92}
H.~Nagamochi and T.~Ibaraki.
\newblock A linear-time algorithm for finding a sparse $k$-connected spanning
  subgraph of a $k$-connected graph.
\newblock {\em Algorithmica}, 7(5{\&}6):583--596, 1992.

\bibitem{kECC:NW}
H.~Nagamochi and T.~Watanabe.
\newblock Computing k-edge-connected components of a multigraph.
\newblock {\em IEICE Transactions on Fundamentals of Electronics,
  Communications and Computer Sciences}, E76--A(4):513--517, 1993.

\bibitem{Tarjan72}
R.~E. Tarjan.
\newblock Depth-first search and linear graph algorithms.
\newblock {\em {SIAM} Journal on Computing}, 1(2):146--160, 1972.

\bibitem{st:t}
R.~E. Tarjan.
\newblock Edge-disjoint spanning trees and depth-first search.
\newblock {\em Acta Informatica}, 6(2):171--85, 1976.

\bibitem{Thorup07}
M.~Thorup.
\newblock Fully-dynamic min-cut.
\newblock {\em Combinatorica}, 27(1):91--127, 2007.
\newblock Announced at STOC'01.

\end{thebibliography}
\pagebreak
\end{document}